\documentclass[conference]{IEEEtran}
\IEEEoverridecommandlockouts
\usepackage{makecell}
\usepackage{bbm}
\usepackage{booktabs}
\usepackage[table,xcdraw]{xcolor}
\usepackage{algorithmic}
\usepackage[linesnumbered,ruled]{algorithm2e}
\usepackage{subfigure}
\usepackage{amsmath}
\usepackage{mathrsfs}
\usepackage{amsfonts}
\usepackage{amsthm}
\newtheorem{theorem}{Theorem}
\newtheorem{definition}{Definition} 
\newtheorem{lemma}{Lemma}

\usepackage{setspace}
\usepackage{verbatim}
\usepackage{datetime}
\usepackage{advdate}
\usepackage{graphicx}
\usepackage{float}
\usepackage{color}
\usepackage{cite}
\usepackage{tikz}

\newcommand{\circlednum}[1]{%
  \tikz[baseline=(char.base)]{
    \node[shape=circle,fill=black,text=white,inner sep=1pt] (char) {#1};
  }%
}

\def\BibTeX{{\rm B\kern-.05em{\sc i\kern-.025em b}\kern-.08em
    T\kern-.1667em\lower.7ex\hbox{E}\kern-.125emX}}
\begin{document}

\title{BEVCooper: Accurate and Communication-Efficient Bird's-Eye-View Perception in Vehicular Networks}

	\author{Jiawei Hou, Peng Yang,,~Xiangxiang Dai,~Mingliu Liu,~Conghao Zhou}
	\author{\IEEEauthorblockN{Jiawei~Hou\IEEEauthorrefmark{1},~Peng~Yang\IEEEauthorrefmark{1},~Xiangxiang Dai\IEEEauthorrefmark{2},~Mingliu Liu\IEEEauthorrefmark{3}, and~Conghao Zhou\IEEEauthorrefmark{4}}
		\IEEEauthorblockA{\IEEEauthorrefmark{1}School of Electronic Information and Communications, Huazhong University of Science and Technology, Wuhan, China \\
			\IEEEauthorrefmark{2}Department of Computer Science and Engineering, The Chinese University of Hong Kong, Hong Kong, China \\
            \IEEEauthorrefmark{3}The State Grid Hubei Electric Power Research Institute, Wuhan, China \\
            \IEEEauthorrefmark{4}School of Telecommunications Engineering, Xidian University, Xi'an, China \\}
		Email: \IEEEauthorrefmark{1}\{jerry$\_$hou, yangpeng\}@hust.edu.cn, \IEEEauthorrefmark{2}xxdai23@cse.cuhk.edu.hk, \IEEEauthorrefmark{3}liumingliu@whu.edu.cn,
    \IEEEauthorrefmark{4}zhouconghao@xidian.edu.cn
    }

\maketitle

\begin{abstract}
Bird’s-Eye-View (BEV) is critical to connected and automated vehicles (CAVs) as it can provide unified and precise representation of vehicular surroundings. However, quality of the raw sensing data may degrade in occluded or distant regions, undermining the fidelity of constructed BEV map. In this paper, we propose BEVCooper, a novel collaborative perception framework that can guarantee accurate and low-latency BEV map construction. We first define an effective metric to evaluate the utility of BEV features from neighboring CAVs. Then, based on this, we develop an online learning-based collaborative CAV selection strategy that captures the ever-changing BEV feature utility of neighboring vehicles, enabling the ego CAV to prioritize the most valuable sources under bandwidth-constrained vehicle-to-vehicle (V2V) links. Furthermore, we design an adaptive fusion mechanism that optimizes BEV feature compression based on the environment dynamics and real-time V2V channel quality, effectively balancing feature transmission latency and accuracy of the constructed BEV map. Theoretical analysis demonstrates that, BEVCooper achieves asymptotically optimal CAV selection and adaptive feature fusion under dynamic vehicular topology and V2V channel conditions. Extensive experiments on real-world testbed show that, compared with state-of-the-art benchmarks, the proposed BEVCooper enhances BEV perception accuracy by up to $63.18\%$ and reduces end-to-end latency by $67.9\%$, with only $1.8\%$ additional computational overhead.
	\end{abstract}

\section{Introduction}

The market penetration rate of connected and automated vehicles (CAVs) equipped with exterior high-end cameras is experiencing rapid growth  \cite{10787093, vanet3, survey1}. These multi-perspective cameras enable CAVs to construct BEV maps, thereby generating unified, accurate representations of their surroundings \cite{bevfusion, bevsurvey1, bevsurvey2}. However, camera's sensing performance degrades significantly under obstruction or at long distances, compromising the fidelity of the constructed bird's-eye-view (BEV) map. To obtain accurate BEV representation, stand-alone perception, which synthesizes BEV map based on sensing data exclusively from a single CAV, is insufficient \cite{10689455, pacp, luo2025improving}. Collaborative BEV perception, which leverages sensing data from multiple CAVs, has garnered significant attention \cite{huang2023v2x, pradhan2024copilot, 10228934}. As illustrated in Figure \ref{fig:scenario_new}, by allowing an \textbf{\textit{ego CAV}} to request perception messages from its neighboring \textbf{\textit{collaborative CAVs}}, collaborative BEV perception enables more accurate BEV map construction and supports safe driving decisions \cite{cobevt}.

According to the stage at which transmitted perception messages are incorporated into the BEV map construction process, 
BEV perception can be classified into three levels: raw-data level \cite{emp, cp2, 10621158}, intermediate feature level \cite{v2vnet, cmass}, and result level \cite{chan2025energy, liu2021livemap}. Among these, intermediate feature level collaboration, which involves sharing locally extracted compact features, offers a promising trade-off between communication efficiency and preservation of BEV-relevant information. Consequently, it has received considerable attention in recent studies \cite{cp1}. Although transmitting intermediate data incurs at least one order of magnitude less communication overhead than raw sensing data (e.g., reducing transmitted data from MBs to KBs \cite{chen2019f}), the spectral resources available for inter-CAV data exchange remain limited. For instance, vehicle-to-vehicle (V2V) links are only allocated a $10$ MHz frequency band at $5.9$ GHz for C-V2X communications in China \cite{5gaa2021deployment}. This scarce bandwidth limits the data transmission rate and makes it impractical for ego CAV to request BEV features from all surrounding CAVs, thereby calling for solutions to the following \textbf{research problems}.

\begin{figure}
    \centering
    \includegraphics[width=0.8\columnwidth]{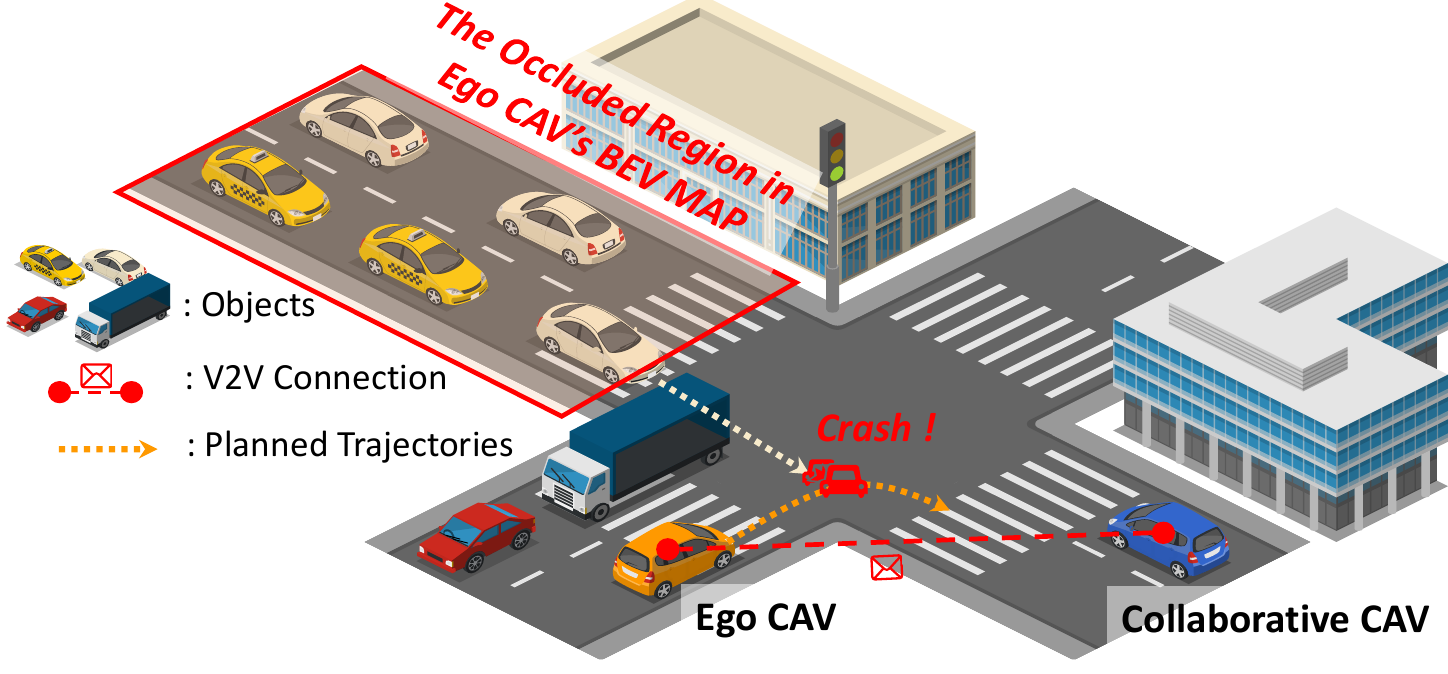}
    \vspace{-0.1in}
    \caption{An illustration of collaborative BEV perception.}
    \label{fig:scenario_new}
    \vspace{-0.24in}
\end{figure}

\circlednum{\scriptsize 1} \textit{How can the ego CAV select an optimal set of collaborative CAVs to construct an accurate BEV map?} Owing to vehicular mobility, collaborative CAVs provide continuously-evolving and varying levels of contribution to the ego CAV’s BEV map construction. Under a limited selection budget, the ego CAV must identify the most valuable collaborators by evaluating their BEV feature utility in real time.

\circlednum{\scriptsize 2} \textit{How can the ego CAV ensure timely BEV map construction in the presence of the straggler effect induced by heterogeneous V2V link quality?} In collaborative BEV perception, the ego CAV cannot initiate data fusion and BEV map construction until it has received requested features from all collaborative CAVs. However, those CAVs with poor V2V link quality can inflate the overall feature transmission latency up to second-level \cite{harbor}, severely impeding real-time BEV map updates. This is inevitable in practice due to frequent signal blockages, dynamic inter-CAV distance variations, etc \cite{boban2010impact}.

Unfortunately, existing approaches face \textbf{fundamental challenges} in tackling above problems. First, although some studies have explored collaborative CAV selection, they typically rely on static sensor metadata, such as camera coverage, to assess utility \cite{wang2024edge, jiawei, mass}. Such metrics are agnostic to the semantic content of the extracted features and cannot reflect their actual contribution to BEV map quality. Furthermore, as the available collaborative CAVs are constantly moving, the utility estimates of collaborative CAVs quickly become outdated. This necessitates a proper balance between exploring new collaborators and exploiting previously inferred utilities. Second, while prior works \cite{pacp, harbor} have explored mitigating straggler effect through adaptive data compressing, they overlook the varying urgency of BEV map construction across different driving scenarios. For instance, ego CAV navigating through high-mobility urban intersections requires rapid BEV updates, whereas one cruising in a stable platoon can tolerate higher latency. Compression schemes that ignore such driving volatility may apply overly aggressive compression in low-urgency settings, degrading accuracy, or insufficient compression in time-sensitive contexts, resulting in excessive delay.

\begin{figure}[t]
		\centering
	\includegraphics[width=0.8\columnwidth]{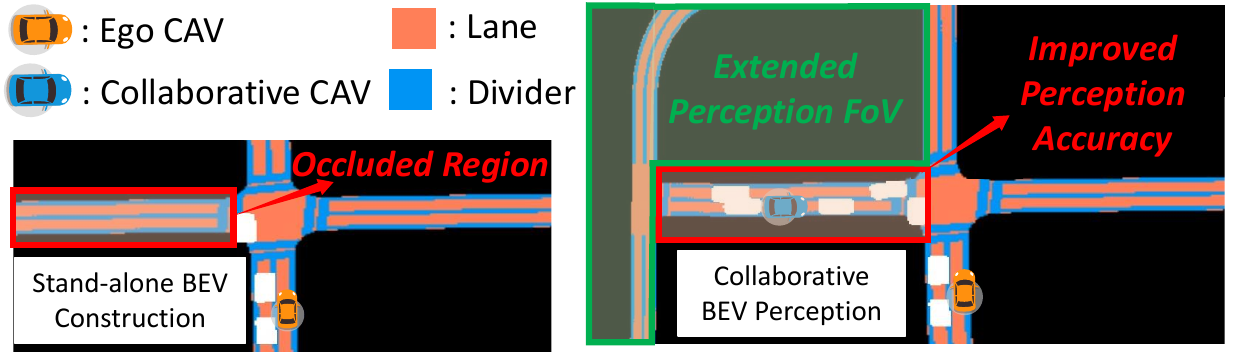}
    \vspace{-0.1in}
            \caption{Benefits of collaborative perception based on BEV.}
            \label{fig:cp}
        \vspace{-0.28in}
\end{figure}

To address these challenges, our preliminary studies identify three critical requirements for accurate and communication-efficient collaborative BEV perception:  \circlednum{\scriptsize 1} \textbf{effective BEV feature utility evaluation, }\circlednum{\scriptsize 2}\textbf{ proper exploration-exploitation in collaborative CAV selection,} and \circlednum{\scriptsize 3} \textbf{driving volatility-aware straggler effect mitigation.} Based on these insights, we propose BEVCooper, a collaborative perception framework that enables ego CAV to adapt to varying environments while maintaining accurate and timely BEV map construction, through the following \textbf{designs and contributions}.

First, we propose a novel BEV feature evaluation metric termed \textbf{\textit{marginal BEV contribution}}, that assesses the incremental improvement in both map accuracy and additional Field-of-View (FoV) provided by a collaborative CAV. This metric enables BEVCooper to precisely identify the most beneficial collaborative CAVs for ego CAV's BEV construction.

Second, we develop an online learning-based CAV selection strategy with an alternating exploration-exploitation architecture. This enables BEVCooper to optimally leverage known high-performance collaborators while systematically evaluating promising but underutilized CAVs. By dynamically adjusting the exploration-exploitation balance under a constrained selection budget, BEVCooper maintains superior BEV perception quality amid continuous vehicular mobility.

Third, we design a driving volatility-aware BEV feature fusion mechanism that dynamically optimizes compression ratios based on both environmental volatility and real-time V2V link quality. Unlike existing approaches that treat straggler mitigation statically, our design enables BEVCooper to adaptively balance feature quality and transmission latency, ensuring timely BEV map construction while maintaining perception accuracy across diverse driving scenarios.

Theoretical analysis shows that, BEVCooper achieves asymptotically optimal CAV selection and adaptive feature fusion in vehicular networks with continuously-changing feature utilities and V2V channel quality. Furthermore, BEVCooper is implemented on real-world platforms, i.e., NVIDIA Jetson Orin and RTX 3080Ti. Extensive experimental results demonstrate BEVCooper's superiority over state-of-the-art methods, achieving improvements of $63.18\%$ in BEV perception accuracy and $67.9\%$ in transmission latency reduction across diverse driving scenarios.

\section{Observations and Motivations}
\label{sec:motivation}
This section presents the motivation for the design of BEVCooper, supported by preliminary experimental analysis.

\subsection{An Unified Metric for BEV Feature Utility Assessment}
\label{sec:motivation_1}

\textbf{Observation:}  We first visualize the BEV map constructed by the ego CAV in the example shown in Figure \ref{fig:scenario_new}. As illustrated in Figure  \ref{fig:cp}, collaborative perception enhances the ego CAV’s perception accuracy within its own FoV and extends its perceptual coverage by incorporating complementary viewpoints, resulting in a more accurate and holistic BEV map. 

\textbf{Motivation:} Therefore, \textit{when quantifying the extent to which a collaborative CAV’s BEV feature enhances the ego CAV’s perception capability, both the improvement in the ego CAV’s perception accuracy and the expansion of its FoV should be simultaneously taken into consideration}. Focusing exclusively on one dimension, such as camera coverage \cite{wang2024edge} or accuracy improvement \cite{jiawei, mass} alone, may overlook valuable data from CAVs with complementary sensing geometries or higher-quality features in specific regions, ultimately compromising the overall quality of the constructed BEV map.

\begin{figure}[t]
		\centering
        		\subfigure[Dynamic BEV contributions]{
		\includegraphics[width=0.48\columnwidth]{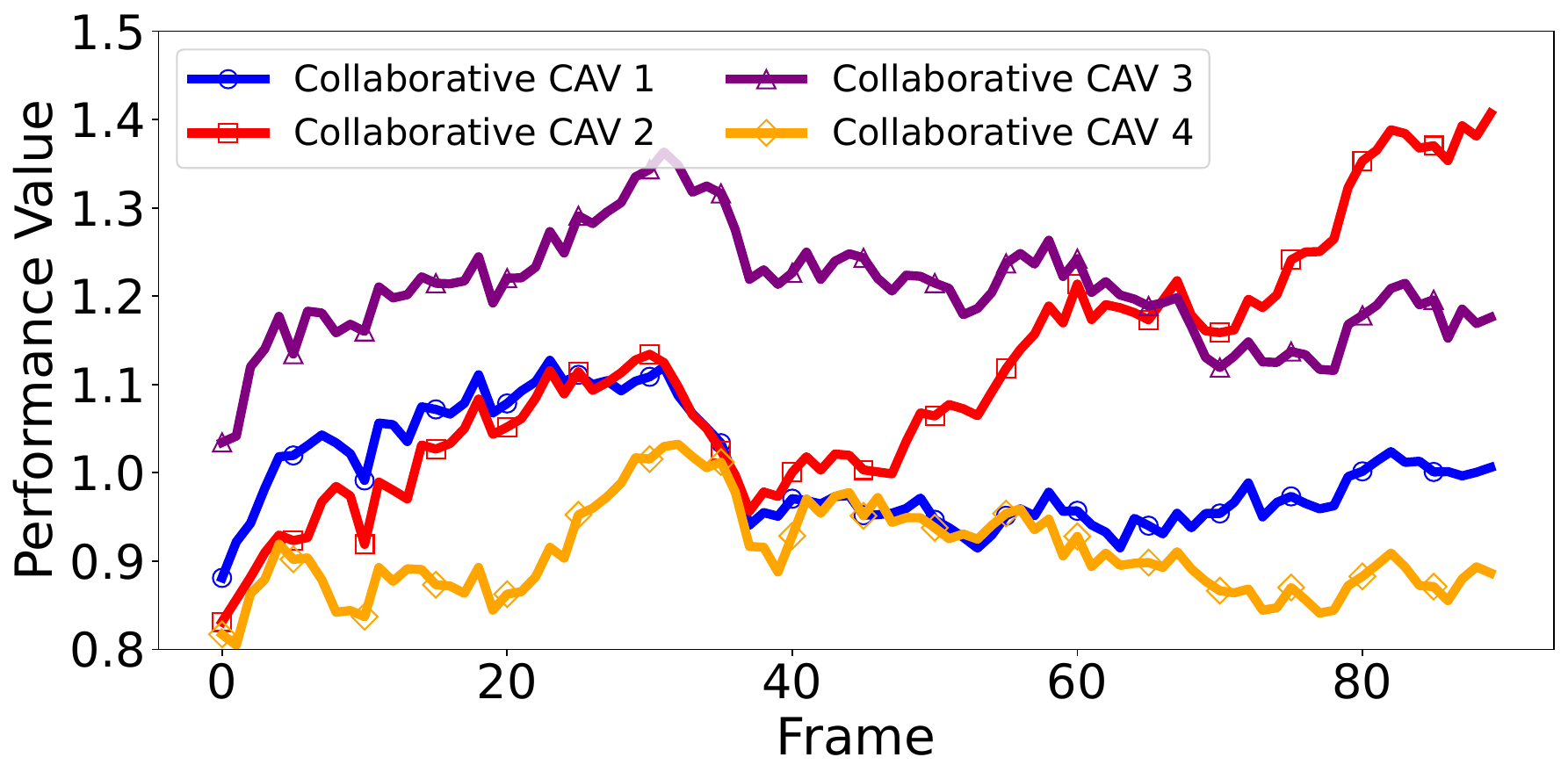}\label{fig:dynamic2}}
        \hspace{-0.05in}
        		\subfigure[Camera content variations]{
	\includegraphics[width=0.48\columnwidth]{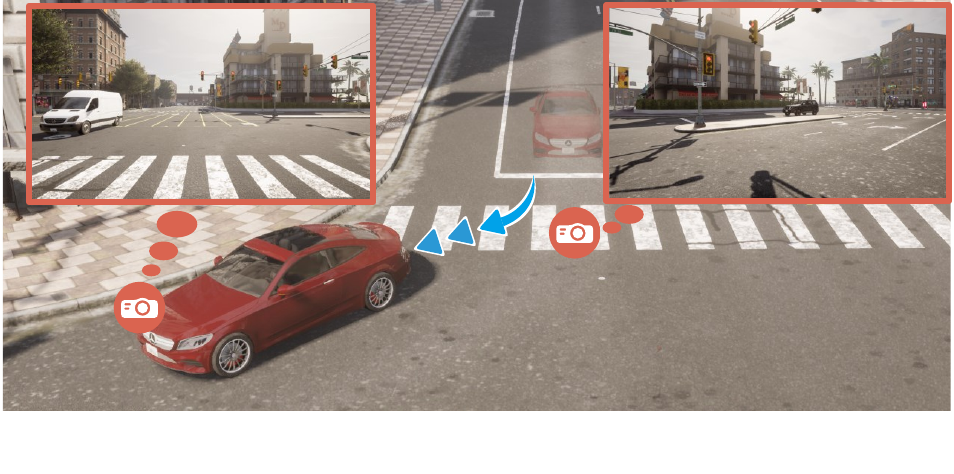}\label{fig:dynamic3}}
            \vspace{-0.2in}
            \caption{Illustrations of the dynamic nature of driving environment.
            }
		\label{fig:moti2}
               \vspace{-0.28in}
\end{figure}
\begin{comment}
\begin{figure}[t]
		\centering
	\includegraphics[width=0.47\columnwidth]{figures/marginal_dynamic_ious_2.pdf}\label{fig:dynamic2}
            \caption{Illustration of the dynamic collaborative perception contributions.}
		\label{fig:moti2}
        \vspace{-0.13in}
\end{figure}
\end{comment}

\subsection{Dynamic Perception Contribution of Collaborative CAVs}
\label{sec:motivation_2}
%\textcolor{red}{motivation 1: dynamic and unpredictable perception contribution of onroad collaborative CAVs.} preliminary experiment conducted on the OPV2V dataset \cite{opv2v} with one CAV clusters 

\textbf{Observation:} For ego CAV, the most intuitive strategy to maximize resource efficiency and enhance its perception accuracy would be to pre-identify and consistently select the CAVs with the highest marginal BEV perception. However, identifying such high-contributing CAVs is non-trivial. To illustrate this, we randomly assign one ego CAV and record the marginal BEV contributions of the remaining CAVs on a simulated dataset \cite{opv2v}. As shown in Figure \ref{fig:dynamic2}, these contributions exhibit significant temporal fluctuations and unpredictability. This variability arises from the rapidly changing driving environment, as depicted in Figure \ref{fig:dynamic3}, where collaborative CAVs frequently relocate to positions with varying perceptual value.

%The detailed calculation procedure is presented in Section \ref{sec:bev_contribution}. 

\textbf{Motivation:} Based on the above observation, relying on pre-determined set of collaborative CAVs \cite{emp, ruiqi1, robust} proves insufficient. \textit{This motivates our online learning-based CAV selection strategy, which dynamically evaluates and selects collaborative CAVs based on their real-time and historical marginal contributions to BEV map construction.}

\subsection{Straggler Effect in Collaborative BEV Perception}
\label{subsubsec:moti2}

\textbf{Observation:} Another critical challenge in collaborative perception is the \textbf{straggler effect}, where excessive feature transmission latency causes the constructed BEV map to deviate from the actual driving environment. To investigate its impact, we evaluate collaborative perception performance under real-world network conditions. Specifically, we adopt CoBEVT \cite{cobevt} as the segmentation model for the BEV map segmentation task. Following the 5G NR V2X sidelink standard \cite{sidelink}, a total data rate of 40–50 Mbps is allocated to collaborative CAVs based on their distances to the ego CAV. The accuracy of a BEV map constructed with latency $x$ ms is quantified by the mean Intersection over Union (mIoU) gap, computed against the ground-truth label from the frame $\lceil \frac{x}{100} \rceil$ steps ahead. As shown in Figure \ref{fig:straggler}, the straggler effect significantly prolongs feature transmission time, reducing perception accuracy by up to $57.8\%$. Furthermore, Figure \ref{fig:straggler_bar} illustrates that higher environmental volatility leads to a larger mIoU gap, thereby exacerbating the challenge of mitigating straggler effects under dynamic conditions.

\textbf{Motivation:} Existing studies has either overlooked the straggler effect \cite{v2vnet, spatio, shao}, or addressed it without accounting for the ego CAV's dynamic driving environment \cite{pacp, v2vnet, harbor}, leading to compromised collaborative BEV perception performance. \textit{This inspires us to integrate driving volatility awareness with straggler mitigation to ensure timely and accurate BEV map construction.}

\section{System Overview and Problem Formulation}

\begin{figure}[t]
		\centering	
        		\subfigure[]{
		\includegraphics[width=0.48\columnwidth]{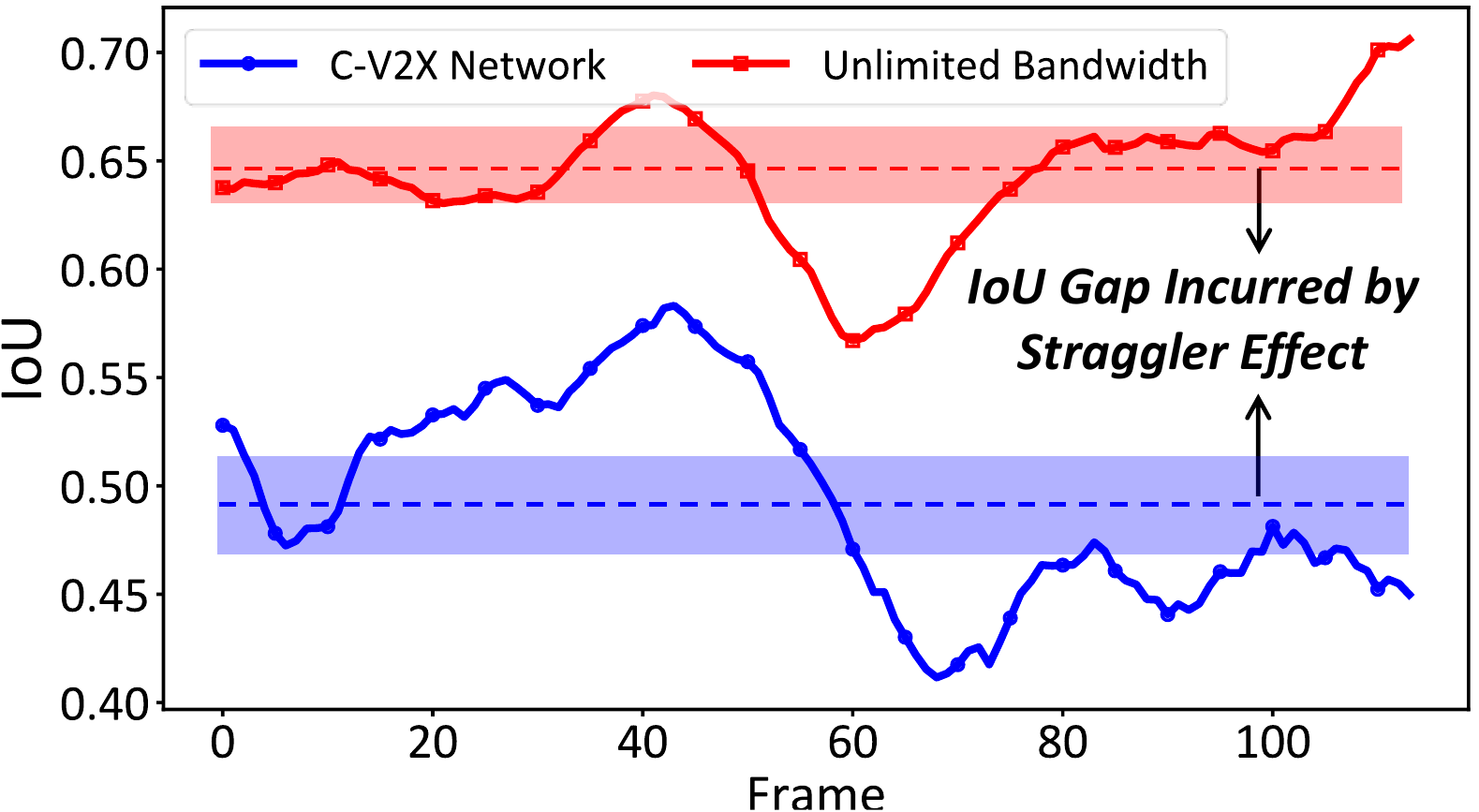}\label{fig:straggler}}
        \hspace{-0.05in}
        		\subfigure[]{
		\includegraphics[width=0.47\columnwidth]{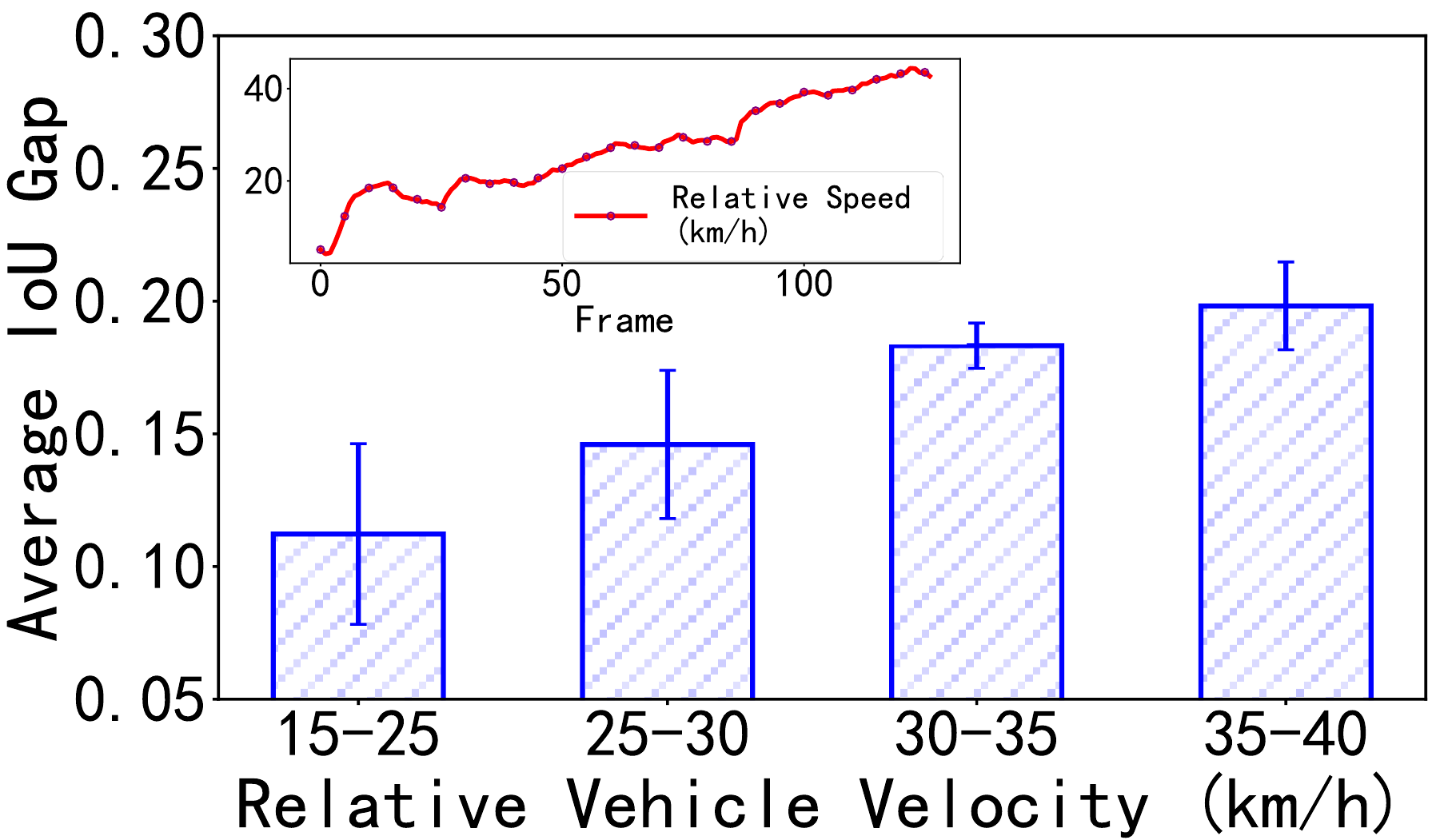}\label{fig:straggler_bar}}
            \vspace{-0.22in}
         \caption{Impact of the straggler effect in collaborative BEV perception.
         }
		\label{fig:moti3}
        \vspace{-0.28in}
\end{figure}

\begin{comment}
\subfigure[]{
		\includegraphics[height=1in]{figures/dynamics3.pdf}\label{fig:dynamic}}
\end{comment}

As shown in Figure \ref{fig:workflow}, we consider an ego CAV driving across the busy urban intersection, where the driving environment may include dynamic objects such as pedestrians and surrounding vehicles. The ego CAV establishes V2V connections with $N$ collaborative CAVs within a stable vehicle cluster that are willing to participate in collaborative perception \cite{cluster}. Each CAV is equipped with four cameras and generates images at $10$ FPS. We consider homogeneous computing resources, i.e., all CAVs have identical computational capabilities \cite{10643366}, and focus on the straggler effect arising from heterogeneous V2V channel quality. Without loss of generality, we assume that all sensors are well-synchronized and share a common clock \cite{opv2v}. Time is divided into discrete slots of duration $\Delta t = 100$ms, aligned with the camera's sensing interval.

%serving as the basic operational units for BEVCooper.

Figure \ref{fig:workflow} illustrates the workflow of BEVCooper in both control plane and data plane. At the beginning of each time slot, the control plane on the ego CAV determines the CAV selection strategy based on historical data and computes the fusion deadline by considering both driving volatility and V2V channel conditions. On the data plane, the selected collaborative CAVs, represented by the set $\mathcal{K}$ with $|\mathcal{K}| = K \le N$, transmit their locally extracted BEV features to the ego CAV. To meet the fusion deadline, each selected CAV adaptively adjusts its compression rate to ensure timely delivery. Upon receiving the features, the ego CAV performs feature fusion and BEV map construction, and subsequently updates the estimated utility of each collaborator based on its contribution to the current BEV map. This updated utility informs the CAV selection strategy in the subsequent time slots. To optimize this process, the ego CAV must accurately evaluate the marginal perception utility of each collaborator’s BEV feature.

\begin{figure}[t]
		\centering
	\includegraphics[width=0.9\columnwidth]{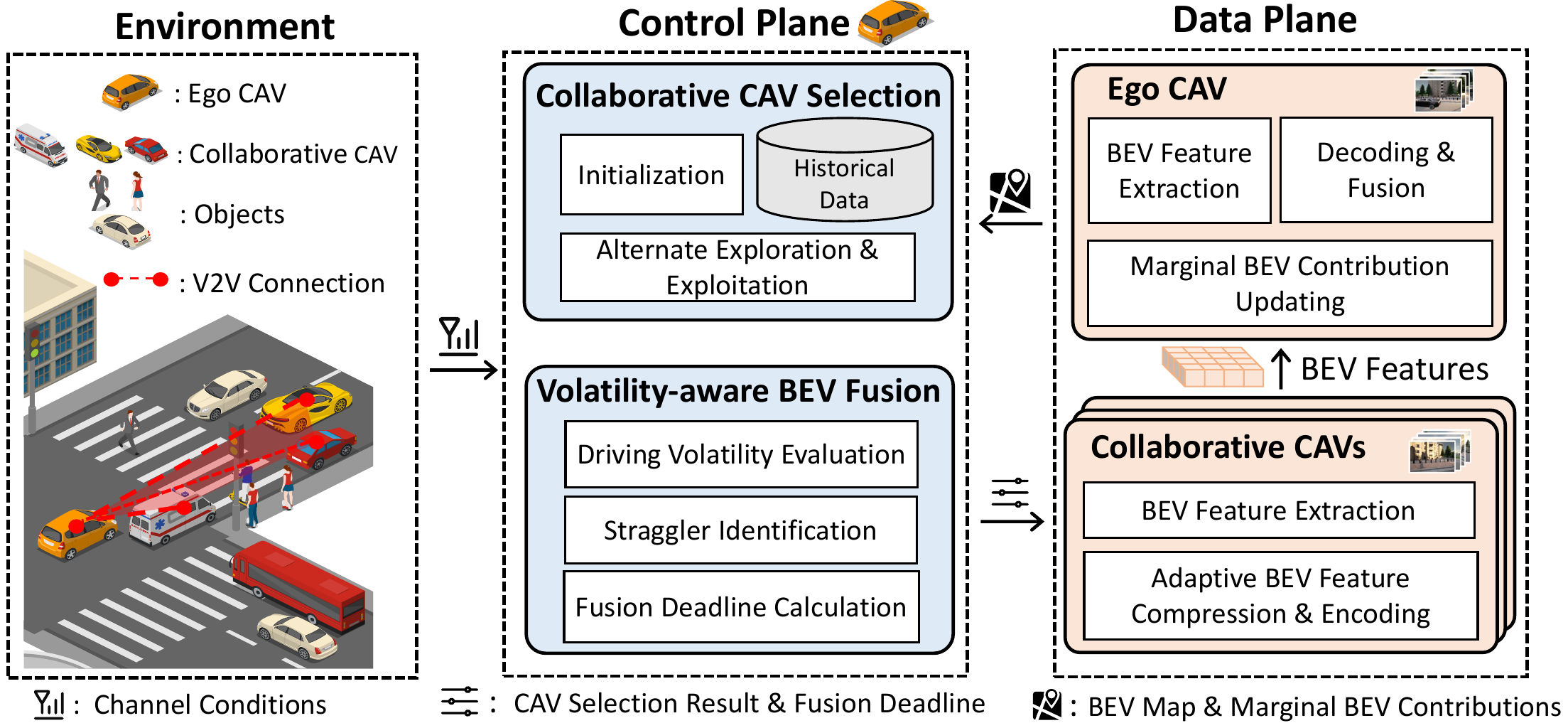}
    \vspace{-0.1in}
            \caption{The workflow of BEVCooper.}
            \label{fig:workflow}
        \vspace{-0.28in}
\end{figure}
    
\subsection{Marginal BEV Contribution}
\label{sec:bev_contribution}
As stated in Section \ref{sec:motivation_1}, collaborative BEV perception provides two principal benefits to ego CAV: enhanced perception accuracy within its existing sensor FoV, and extended perception coverage beyond its FoV. To quantify the utility of each selected CAV’s BEV feature, we introduce the marginal BEV contribution metric. This metric reflects the incremental value contributed by a collaborative CAV in terms of both segmentation accuracy and spatial awareness.

\textbf{Marginal Segmentation Accuracy:} We use the marginal segmentation accuracy to quantify the extent to which the collaborative CAV enhances the perception quality within the ego CAV's FoV. In specific, the marginal segmentation accuracy $m_i$ of CAV index $i$ is computed as:

        \begin{equation}
             m_i= 1 - \text{IoU}\left(\text{BEV}_{\mathcal{K}}, \text{BEV}_{\mathcal{K}\backslash\{i\}}\right),
             \label{eq:msa}
        \end{equation}
where $\text{BEV}_{\mathcal{K}}$ and $\text{BEV}_{\mathcal{K}\backslash\{i\}}$ denote the segmentation output with perception data from CAV set $\mathcal{K}$ and $\mathcal{K}\backslash\{i\}$, respectively. Function $\text{IoU}(\cdot)$ calculates the mean IoU between two segmented BEV map.  Notably, this metric can be calculated without requiring ground truth segmentation results.

\textbf{Normalized Extended FoV:} In addition to improving BEV segmentation accuracy within the ego CAV's FoV, collaborative perception also offers the advantage of extending the ego CAV’s FoV, allowing it to observe more distant or occluded objects. To measure this extended coverage, we define a normalized metric $A_i$, representing the additional FoV area contributed by selected CAV $i$:

\begin{equation}
    A_i = 1  - \left(\frac{|\text{FoV}_i \cap \text{FoV}_e|}{|\text{FoV}_i|}\right),
    \label{eq:area}
\end{equation}
where $\text{FoV}_e$ and $\text{FoV}_i$ denote the FoVs of the ego CAV and CAV $i$, respectively, function $|\cdot|$ denotes the area operator. As shown in Figure \ref{fig:ee1}, each CAV's perception FoV is modeled as a rectangle on the BEV map, which can be computed from metadata such as GPS coordinates and vehicle orientation.

Equation (\ref{eq:area}) quantifies the incremental FoV of a collaborative CAV by subtracting the overlapping polygonal area shared with the ego CAV’s FoV. A higher $A_i$ indicates a greater portion of newly observable area beyond the ego’s original FoV enabled by CAV $i$. Although selecting collaborative CAVs with highly overlapping perception regions may yield higher marginal segmentation accuracy, it limits the benefit of extended FoV. To balance this trade-off, we define a unified contribution score that integrates both metrics:

\begin{algorithm}[t]
		\caption{Online Collaborative CAV Selection}\label{alg:aecs}
		 \textbf{Data:} $N$, $K$, $\Theta(t)$\,\, \textbf{Result:} $\mathbf{a}(t), 1\le t \le T$ \;
            $I_i\gets 1, \bar{g}_i \gets 0, O_i \gets 1 $ for all $i = 1, \dots, N$\;
            \For{$t = 1$ to $\lceil N/K \rceil$}{Sequentially select $K$ CAVs, initialize $\bar{g}_i$\;}
		\While{$\lceil N/K \rceil < t \leq T $}{
        \eIf{$2^{O_i} -1 < \Theta(t)$}{
         Partition the CAV candidate set into $\lceil N / K \rceil$ groups with an interval of $K$ per group\;
         \If{($N \bmod K) > 0$}{Assign ($N \bmod K$) CAVs with largest $\bar{g}_i$ additionally to the last group;}
         Explore each group for $2^{O_i-1}$ times\;
         $O_i \gets O_i + 1, t \gets t + \lceil N/K \rceil * 2^{O_i-1} $;
         }{
         Exploit top-$K$ CAVs for $2^{I_i-1}$ times\;
         $I_i \gets I_i + 1, t \gets t + 2^{I_i-1} $;
         }
		}
		\vspace{-0.05in}
	\end{algorithm}

\begin{definition}
\label{definition1}
    \textbf{Marginal BEV Contribution.} If a collaborative CAV $i$ is  selected, its marginal BEV contribution to the ego CAV in current collaborative perception round is:
    \begin{equation}
                    g_i = m_i+\omega A_i,
            \label{eq:marginal_gain}
    \end{equation}
    where $\omega \in [0,1]$ is a weighted factor that balances segmentation accuracy and coverage contribution of individual CAVs.
\end{definition}

\noindent\textbf{Remark 1.} The value of $\omega$ can be adaptively adjusted. When the ego CAV already achieves high BEV segmentation accuracy within its own FoV, increasing $\omega$ shifts the selection toward collaborators that provide broader spatial coverage.

\subsection{Problem Formulation}

Let $ a_i(t) \in \{0, 1\} $ denote the action variable indicating whether collaborative CAV $ i \in \{1, \dots, N\} $ is selected at time $ t $, and define the action vector as:
\begin{align}
    \mathbf{a}(t) = \left(a_1(t), a_2(t), \dots, a_N(t)\right) \in \{0,1\}^N, \quad 
\end{align}
where the summation of $a_i(t)$ is subjected to $\sum_{i=1}^N a_i(t) = K$. As the driving environment is constantly-changing, we consider each CAV $i$ is associated with a finite, hidden and dynamic reward $g_i(t) \in \mathcal{G}_i$, which captures its marginal BEV contribution. Then the reward received from selecting CAV $i$ is denoted by: $r_i(t) = g_i(t)$. If CAV $i$ is selected at time $t$, i.e., $a_i(t) = 1$, its reward is accessed and updated. We assume the reward of selected collaborative CAV evolves according to a Markov transition model \cite{xiong2022learning, restless4}:
\begin{equation}
      \mathbb{P}(g_i(t+1) \mid g_i(t), a_i(t) = 1) = P_i(g_i(t), g_i(t+1)).
\end{equation}
If CAV $i$ is not selected, its reward, i.e., marginal BEV contribution, evolves according to an independent and unknown stochastic process due to restless dynamics \cite{dai2024quantifying}. The ego CAV’s objective is to maximize the expected cumulative reward over a finite horizon $T$:
\begin{align}
    \label{eq:objective}\max_{\{\mathbf{a}(t)\}_{t=1}^T} \quad &\mathbb{E} \left[ \sum_{t=1}^T \sum_{i=1}^N a_i(t) \cdot g_i(t) \right],\\
    s.t. & \label{eq:objective_con1}\quad a_i(t) \in \{0, 1\},\\
    & \label{eq:objective_con2}\sum_{i=1}^N a_i(t) = K,\quad \forall t \in \{1, \dots, T\},
\end{align}
where (\ref{eq:objective_con1}) and (\ref{eq:objective_con2}) are constraints enforcing binary decision variables and a limited selection budget. Although the simplest solution for (\ref{eq:objective}) would be to always select the top-$K$ CAVs with the highest $g_i(t)$ values, such a strategy requires full knowledge of each CAV's underlying Markovian transition model, which is inherently stochastic and typically unknown in practice. Moreover, the value of $g_i(t)$ is highly dynamic and and rapidly becomes outdated. Combined with the limited selection budget, the ego CAV must balance the learning of reward distributions from different CAVs and the selection of CAVs that are currently believed to provide the highest reward.

\begin{figure}[t]
		\centering
                		\subfigure[]{
		\includegraphics[height=1.05in]{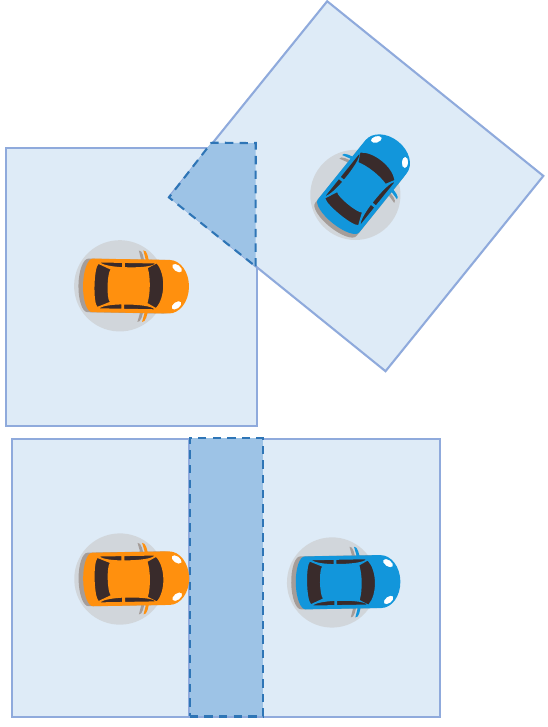}\label{fig:ee1}}
        		\subfigure[]{
		\includegraphics[height=1.05in]{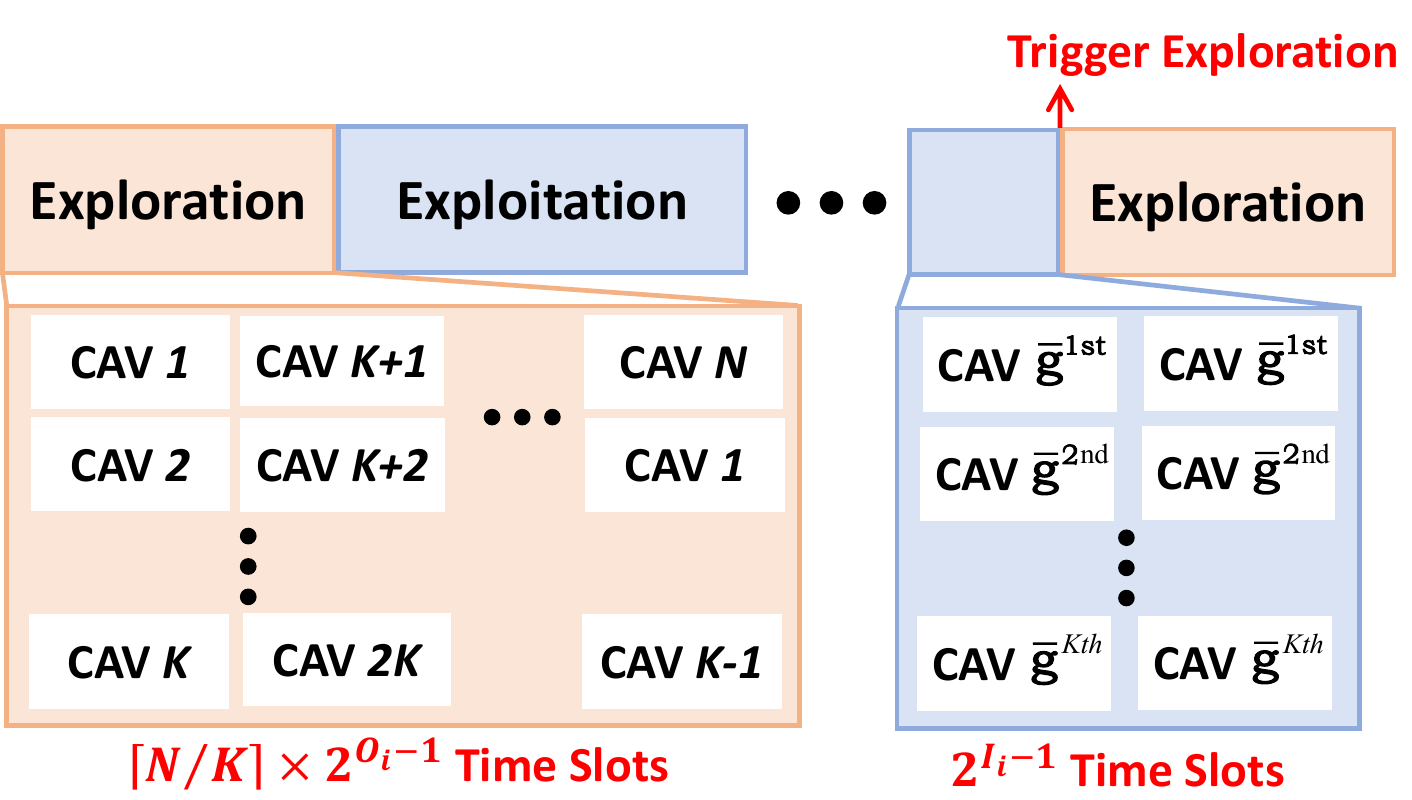}\label{fig:ee2}}
        \vspace{-0.1in}
         \caption{ Illustration of the (a) normalized extended FoV (b) alternate exploration and exploitation structure.}
		\label{fig:ee}
                \vspace{-0.25in}
\end{figure}

\section{Online Collaborative CAV Selection}

Given the complexity of the formulated problem in dynamic environments, BEVCooper incorporates an online collaborative CAV selection algorithm that combines deterministic exploration and exploitation in a cyclic structure.

\subsection{Algorithm Design}
As shown in Algorithm \ref{alg:aecs}, the proposed online collaborative CAV selection comprises two main components.

\textbf{Initialization (lines 1-5).} The ego CAV first initializes the sample mean reward $\bar{g}_i$ for all collaborative CAVs. It then sets phase counters $O_i, I_i$, which respectively track the number of exploration and exploitation phases for each CAV $i$. After that, the ego CAV performs an initial exploration phase over $\lceil N/K \rceil$ time slots, during which $K$ CAVs are selected sequentially at each time step, ensuring that a minimal amount of information is obtained about all CAVs at the start.

\textbf{Alternate Exploration and Exploitation (lines 6-18).} 
To address the dynamic and unknown rewards, the algorithm proceeds by alternating between exploration and exploitation phases. A threshold-driven trigger mechanism is designed to control the switch between two phases: If the number of time slots spent on exploration phases is below the current threshold $\Theta(t)$, e.g., $2^{O_i} -1 < \Theta(t)$. An exploration phase, in which each CAV is selected at least $2^{O_i-1}$ times, will be triggered. $\Theta(t)$ is a predefined function that grows logarithmic with time slot $t$. If the exploration phase is not triggered, the algorithm enters an exploitation phase. The ego CAV continuously selects top-$K$ collaborative CAVs for $2^{I_i-1}$ times. The exponential terms $2^{O_i-1}$ and $2^{I_i-1}$ ensure sufficient sampling while adapting to each vehicle's historical performance.

The alternate exploration and exploitation structure of Algorithm \ref{alg:aecs} is illustrated in Figure \ref{fig:ee2}. This design strikes a balance between maintaining accurate CAV perception quality estimates and maximizing cumulative perception reward in restless changing urban driving environment. In addition, the exponential growth of exploration and exploitation phases, based on counters $O_i$ and $I_i$, ensures a balance between sampling frequency and computational efficiency. While the duration of a CAV cluster is inherently limited in real-world driving scenarios, substantial changes in cluster structure, e.g., collaborative CAVs joining/leaving, can be addressed using existing methods \cite{cluster, huang2018path}, which re-cluster the CAVs and reset the corresponding counters $O_i$ and $I_i$ accordingly.

\begin{table}[]
\centering
\caption{Comparison of execution time of two modules on different devices. Figures in () indicate the running FPS.}
\label{tab:complexity}
\vspace{-0.08in}
\resizebox{\columnwidth}{!}{%
\begin{tabular}{|c|c|c|}
\hline
                       & Jetson Orin  & RTX 3080Ti \\ \hline
BEV Feature Extraction &    425.7 ms (2.35) $\pm $3         &  8.5 ms (118) $\pm$ 0.3         \\ \hline
Segmentation Head      &     3.84 ms (260) $\pm$ 0.15        &   2.03 ms (492) $\pm$ 0.04         \\ \hline
\end{tabular}%
}
        \vspace{-0.25in}
\end{table}

\subsection{Algorithm Analysis}
\textbf{Complexity Analysis.} Compared to vanilla collaborative BEV perception, the additional computational overhead introduced by the Algorithm \ref{alg:aecs} primarily stems from processing the BEV segmentation head $(K+1)$ times to update $\bar{g_i}, i=1,...,K$. However, as shown in Table \ref{tab:complexity}, through practical deployment, the segmentation head \cite{cobevt} executed on the ego CAV incurs only millisecond-level latency, enabling the Algorithm \ref{alg:aecs} to operate in real time.

\textbf{Performance Analysis.} The performance of Algorithm \ref{alg:aecs} is measured by its ability to approach the cumulative perception contribution that could be achieved with full knowledge of the underlying vehicular dynamics. To quantify the performance loss due to uncertainty and learning of the dynamic environment, we introduce the definition of \textbf{\textit{regret}}. Let $\mu_i = \mathbb{E}[g_i(t)]$ be the stationary marginal collaborative perception contribution of CAV $i$, and $\delta$ be a descending permutation of the collaborative CAV set:
    \begin{equation}
    \label{eq:optimal} \underbrace{\mu_{\delta_1}\ge\mu_{\delta_2}\ge\mu_{\delta_3}\ge...\ge\mu_{\delta_K}}_{\text{best possible policy} }\ge ...\ge\mu_{\delta_N},
    \end{equation}
    where  $\delta_i$ denotes the $i$-th vehicle in the descending order of $\mu_i$. The optimal policy in (\ref{eq:optimal}) presumes full prior knowledge of the system dynamics, such as the transition probabilities governing the Markovian evolution of CAV selection rewards. However, in practice, Algorithm \ref{alg:aecs} must learn these dynamics online through interactions with the environment. Consequently, we define the cumulative CAV selection regret as the performance loss incurred due to this lack of prior knowledge:
    \begin{equation}
    \label{eq:reg}
        Reg(T) = T\sum\limits_{i=1}^K\mu_{\delta_i}-\mathbb{E}\left[\sum\limits_{t=1}^T\sum\limits_{i=1}^Na_i(t)g_i(t)\right].
    \end{equation}
Next, we prove that Algorithm \ref{alg:aecs} approaches the optimal policy by bounding the cumulative regret of CAV selection.

\begin{theorem}
\label{math:theorem1}
    Let $D>0$ be a constant and define $\Theta(t) = D\log_2t, t\ge1$. Then, for any time horizon $T$, the cumulative regret $Reg(T)$ of Algorithm \ref{alg:aecs} satisfies: $Reg(T) = O(\log T)$.
\end{theorem}
\begin{proof}
    Please refer to Appendix \ref{sec:appendix1}.
\end{proof}
\noindent\textbf{Remark 2.} The logarithmic-order regret bound in Theorem \ref{math:theorem1} demonstrates that, in restless driving environments, Algorithm 1 can reliably learn and strike an effective balance between exploration and exploitation during collaborative CAV selection. However, prolonged transmission delays caused by poor V2V channel conditions can impede the ego CAV’s BEV map construction and substantially reduce its accuracy.
 \begin{figure}[t]
            \centering
            \includegraphics[width=0.82\columnwidth]{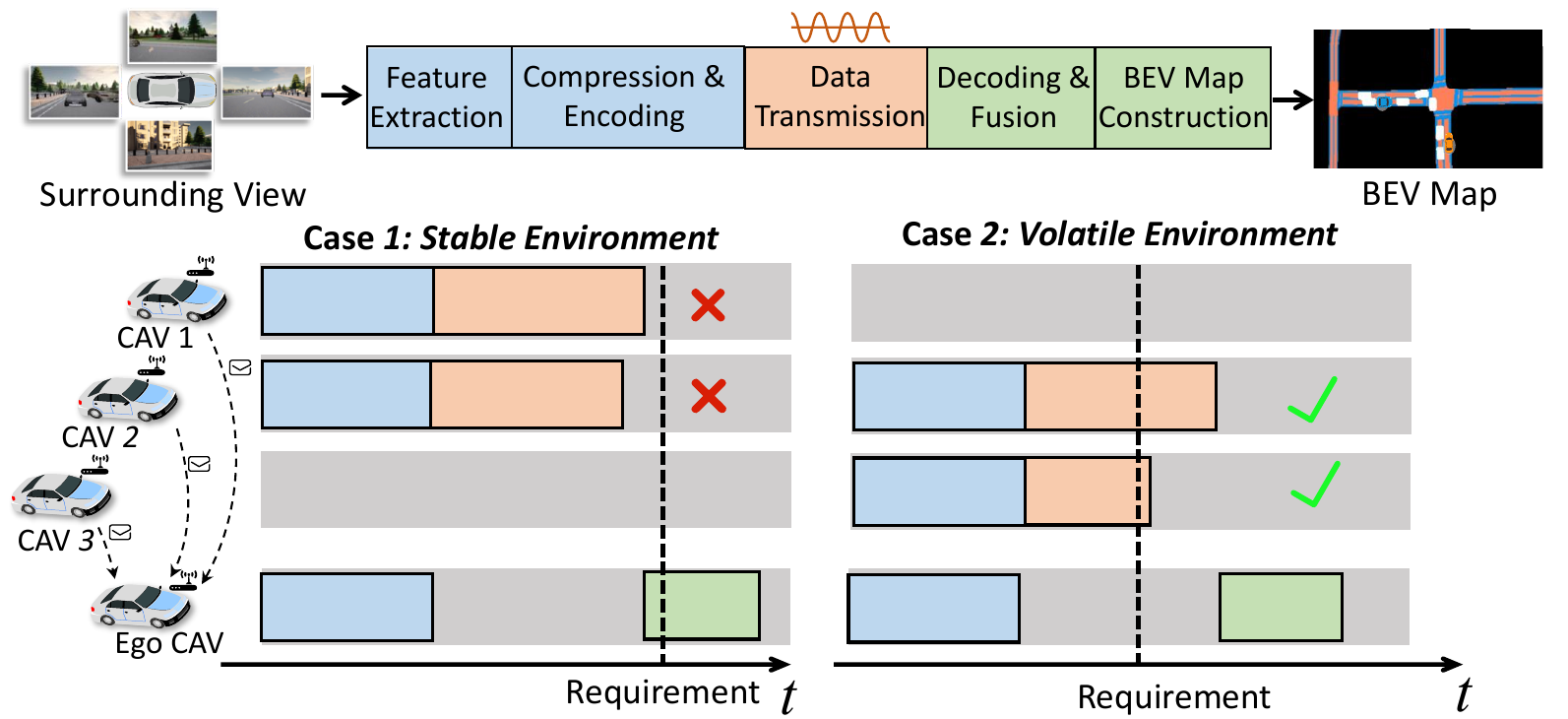}
            \vspace{-0.17in}
            \caption{Perception framework structure breakdown, with \textcolor{green}{$\checkmark$} indicating straggler CAVs and \textcolor{red}{$\times$} indicating non-stragglers. The ego CAV has more stringent latency requirement for BEV map construction in volatile environment.
            }
            \label{fig:frame_structure}
                    \vspace{-0.25in}
        \end{figure}

\section{Volatility-aware BEV Feature Fusion}

To alleviate the straggler effect in collaborative BEV perception, an intuitive approach is to reduce the time required for stragglers to transmit BEV features by feature compression \cite{compress1, cobevt}. this raises two key questions. First, \textbf{how can a straggler be identified?} As illustrated in Figure \ref{fig:frame_structure}, whether a collaborative CAV is considered a straggler depends not only on its channel quality, but also on the latency requirements of the ego CAV. Second, \textbf{how should the compression ratio be determined?} This involves a trade-off: while higher compression reduces transmission delay, it also leads to greater information loss, lowering the utility of the received BEV features. Conversely, lower compression preserves data quality but fails to address straggler-induced latency.

\subsection{Volatility-aware Straggler Identification}

Motivated by the dynamic nature of urban driving environments, as discussed in Section \ref{subsubsec:moti2}, we identify straggler CAVs through a quantitative assessment of driving volatility. To capture the extent to which a vehicle’s driving state diverges from that of its surrounding environment \cite{volatility1, volatility2}, the definition of driving volatility is presented below.

\begin{definition}
\label{definition2}
    \textbf{Driving Volatility.} Given $M>0$ objects within the ego CAV's FoV, driving volatility $v_d$ is quantified as the root mean square of the relative longitude velocity deviations\footnote{This section focuses on the dynamics within a single time slot. For simplicity, the current time slot $t$ is omitted from the notations.}:

    \begin{equation}
        v_d = \sqrt{\frac{1}{M}\sum\limits_{i=1}\limits^{M}\left(v_i-v_e\right)^2},
        \label{eq:volatility}
    \end{equation}
    where $v_e$ and $v_i$ denote the longitude velocity of the ego CAV and $i$-th surrounding object, respectively.
\end{definition}

\noindent\textbf{Remark 3.} A higher driving volatility indicates more drastic changes in the road environment, thereby exacerbating the negative impact of the straggler effect on collaborative perception accuracy. Consequently, the ego CAV has a more urgent demand for fresh BEV features, necessitating a higher compression rate from collaborative CAVs. With $v_d$, stragglers can be identified by setting a fusion deadline.

\begin{definition}
\label{definition3}
    \textbf{Fusion Deadline.} The fusion deadline, denoted $l_f$, specifies the latest time by which collaborative CAVs must deliver their BEV features for fusion. $l_f$ is calculated by:
      \begin{equation}
    \label{eq:fd}
            l_f=l_f^{min}+\left(l_f^{max}-l_f^{min}\right)e^{-\alpha v_d},
    \end{equation}
    where $\alpha > 0$ is a decay constant that balances sensitivity to $v_d$ with deadline flexibility. $l_f^{min}$ and $l_f^{max}$ denote the earliest and latest times at which the ego CAV can initiate fusion, respectively, their values are set based on V2V channel quality.

\end{definition}
\noindent\textbf{Remark 4.} With respect to driving volatility, $l_f$ from  (\ref{eq:fd}) satisfies the following properties: 1) Boundedness: with $v_d \rightarrow 0$, $l_f \rightarrow l_f^{max}$, when $v_d \rightarrow\infty$, we have $l_f \rightarrow l_f^{min}$. 2) Monotonicity: since $\frac{dl_f}{dv_d}<0$, the value of  $l_f$ decreases monotonously with driving volatility. This indicates that the deadline tightens as driving volatility $v_d$ increases. 3) Convexity: calculate the second derivative of $l_f$ with respect to $v_d$, we have $\frac{d^2\, l_f}{dx^2}>0$, which exhibits the convexity of the function and diminishing marginal increase in freshness data demand as driving volatility grows. These properties guarantee that the fusion deadline can adaptively reflect the requirements of the ego CAV. Vehicles that fail to transmit BEV features by $l_f$ are identified as stragglers, enabling the ego CAV to adjust compression ratios dynamically.

\subsection{Adaptive BEV Feature Compression}
To address the second question, we leverage deep neural networks (DNNs) for adaptive BEV feature compression. Let $\rho\ge 1$ denote the BEV feature compression ratio. Since the variation in compression and encoding latency resulting from different compression ratios is negligible compared to the transmission latency induced by fluctuations in V2V channel quality, stragglers primarily control the arrival time of BEV features at the ego CAV by adjusting $\rho$. 

The workflow for mitigating the straggler effect proceeds as follows. Each identified straggler first selects the minimum compression ratio $\rho$ that ensures its transmission completes before the fusion deadline $l_f$. It then compresses its extracted BEV features using a DNN-based encoder trained to minimize task-relevant information loss. Upon reception, the ego CAV reconstructs the features with the goal of maximizing perceptual fidelity. Finally, acknowledging that compression introduces inevitable information loss, the marginal BEV contribution of straggler CAVs is compensated based on the applied compression ratio. Specifically, let $\Delta g_i(\rho)$ denote the degradation in marginal BEV contribution due to compression, which is computed as:
\begin{equation}
\label{eq:compensation}
    \Delta g_i(\rho) = \beta \left(e^{-\gamma\rho_0}-e^{-\gamma\rho}\right),
\end{equation}
where $\rho_0=1$ means no compression, $\beta$ and $\gamma$ are scenario-dependent constants and can be obtained through offline fitting. Using (\ref{eq:compensation}), $g_i$ is updated according to $g_i \leftarrow g_i + \Delta g_i(\rho)$.

\begin{figure*}[t]
		\centering
                		\subfigure[Regret v.s. different scenarios]{
		\includegraphics[width=0.23\textwidth]{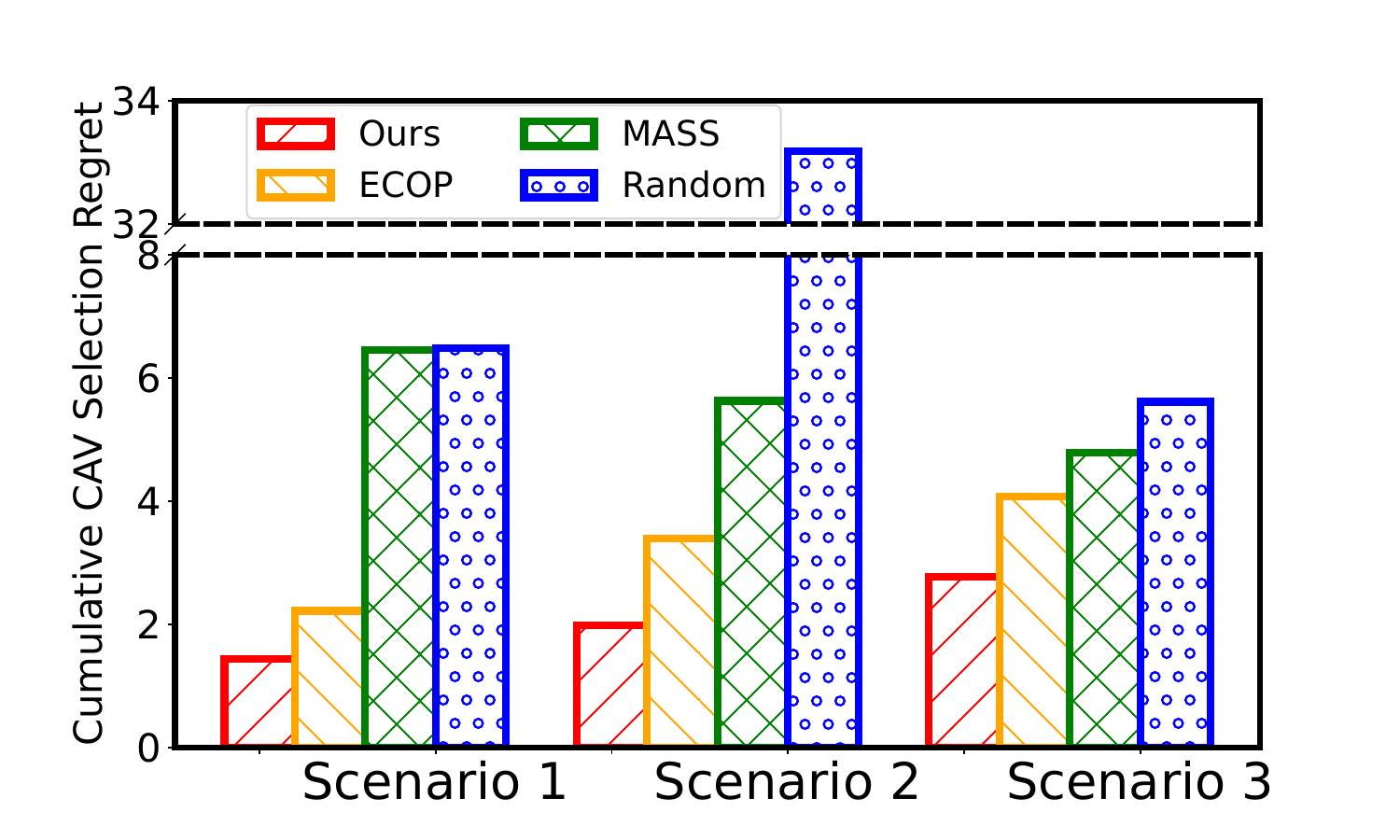}\label{fig:sim1_1}}
            \hspace{-0.15in}
        		\subfigure[The CDF of instantaneous regret]{
		\includegraphics[width=0.24\textwidth]{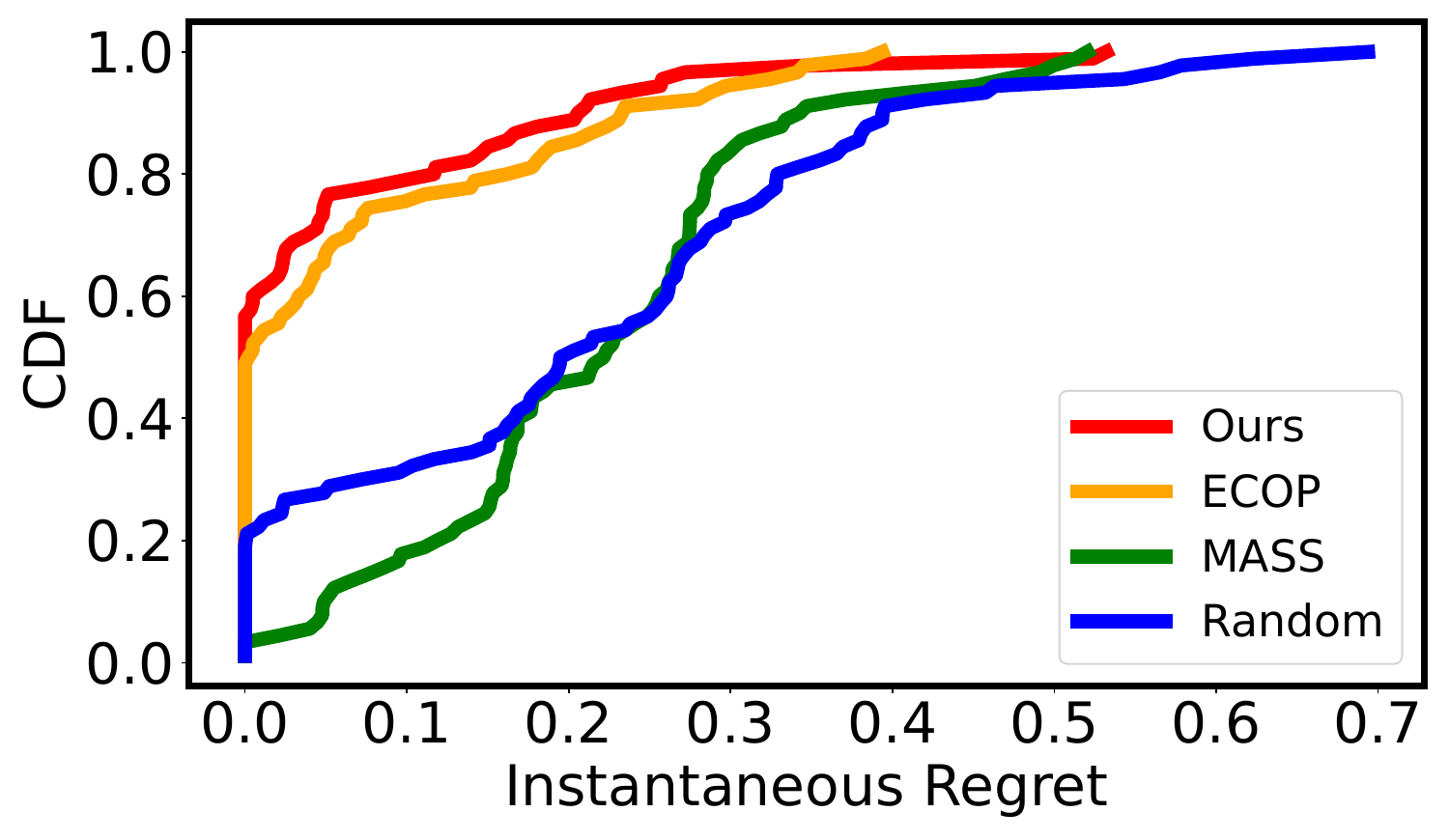}\label{fig:sim1_3}}
          \hspace{-0.07in}
          \subfigure[Regret v.s. CAV selection budget]{
		\includegraphics[width=0.24\textwidth]{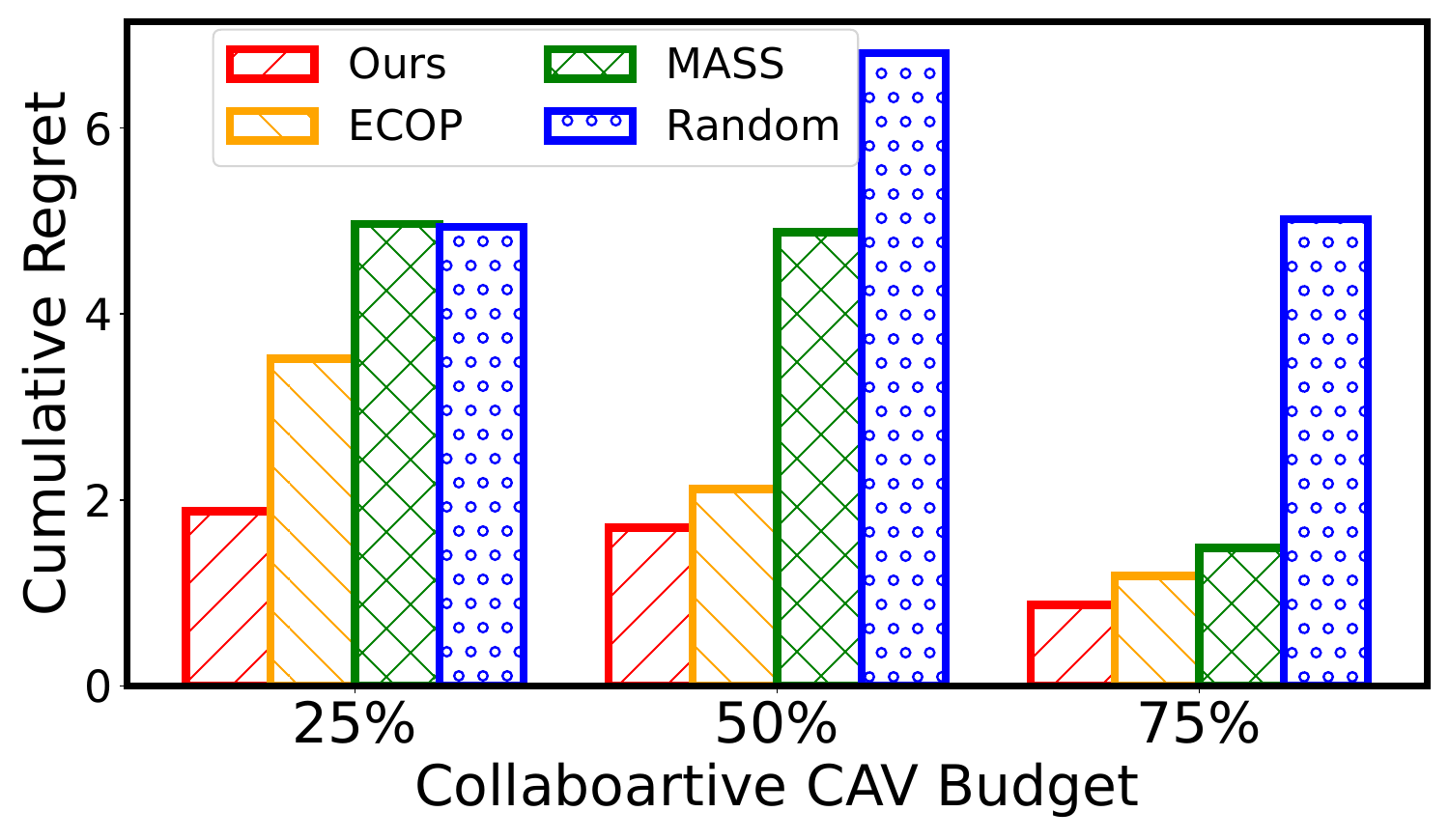}\label{fig:sim1_2}}
         \hspace{-0.07in}
        		\subfigure[Length of phases v.s. value of $D$]{
		\includegraphics[width=0.24\textwidth]{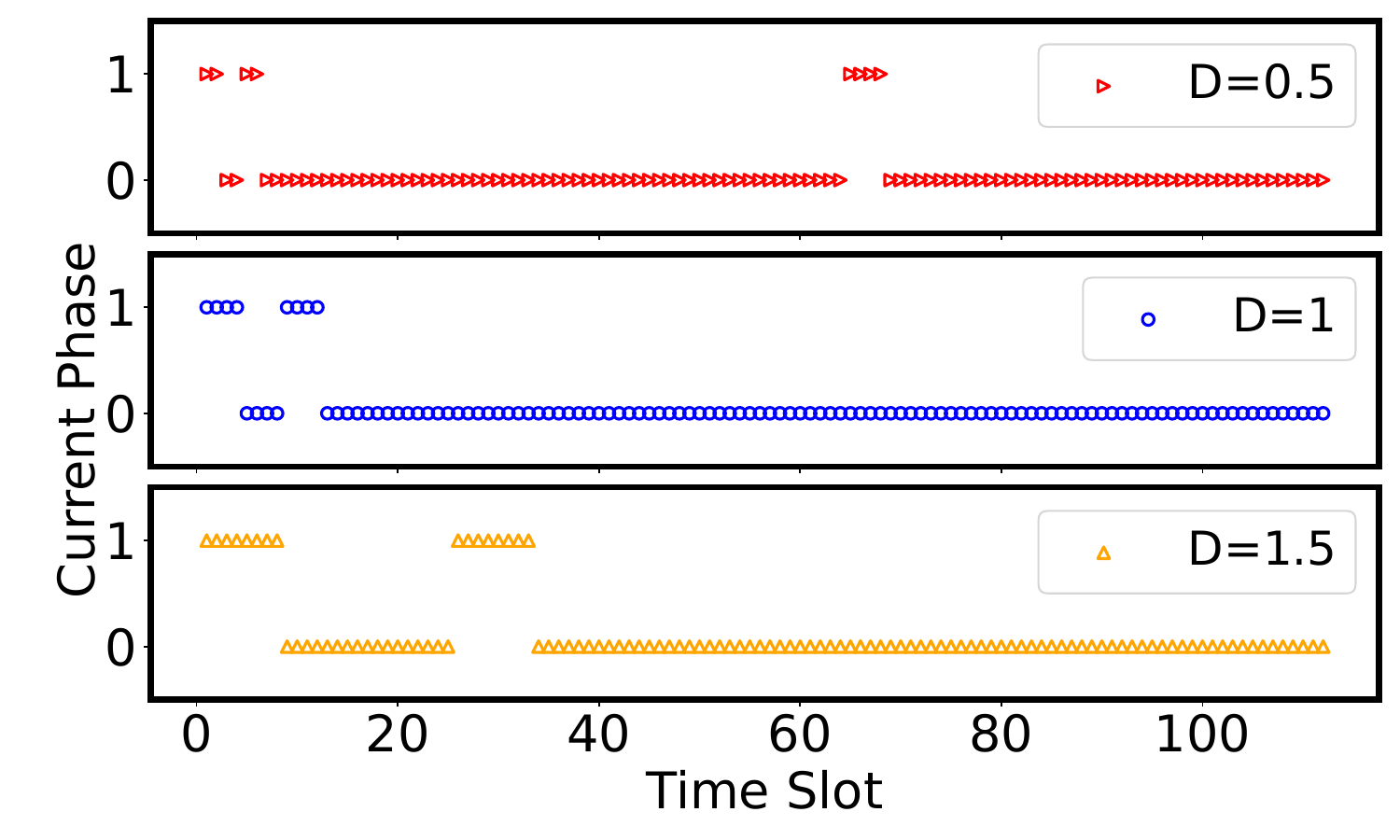}\label{fig:sim1_4}}
        \vspace{-0.1in}
         \caption{Comparative analysis of different collaborative CAV selection methods.
         }
		\label{fig:sim1}
                \vspace{-0.2in}
\end{figure*}

\begin{figure*}[t]
		\centering
                		\subfigure[High V2V throughput]{
		\includegraphics[width=0.23\textwidth]{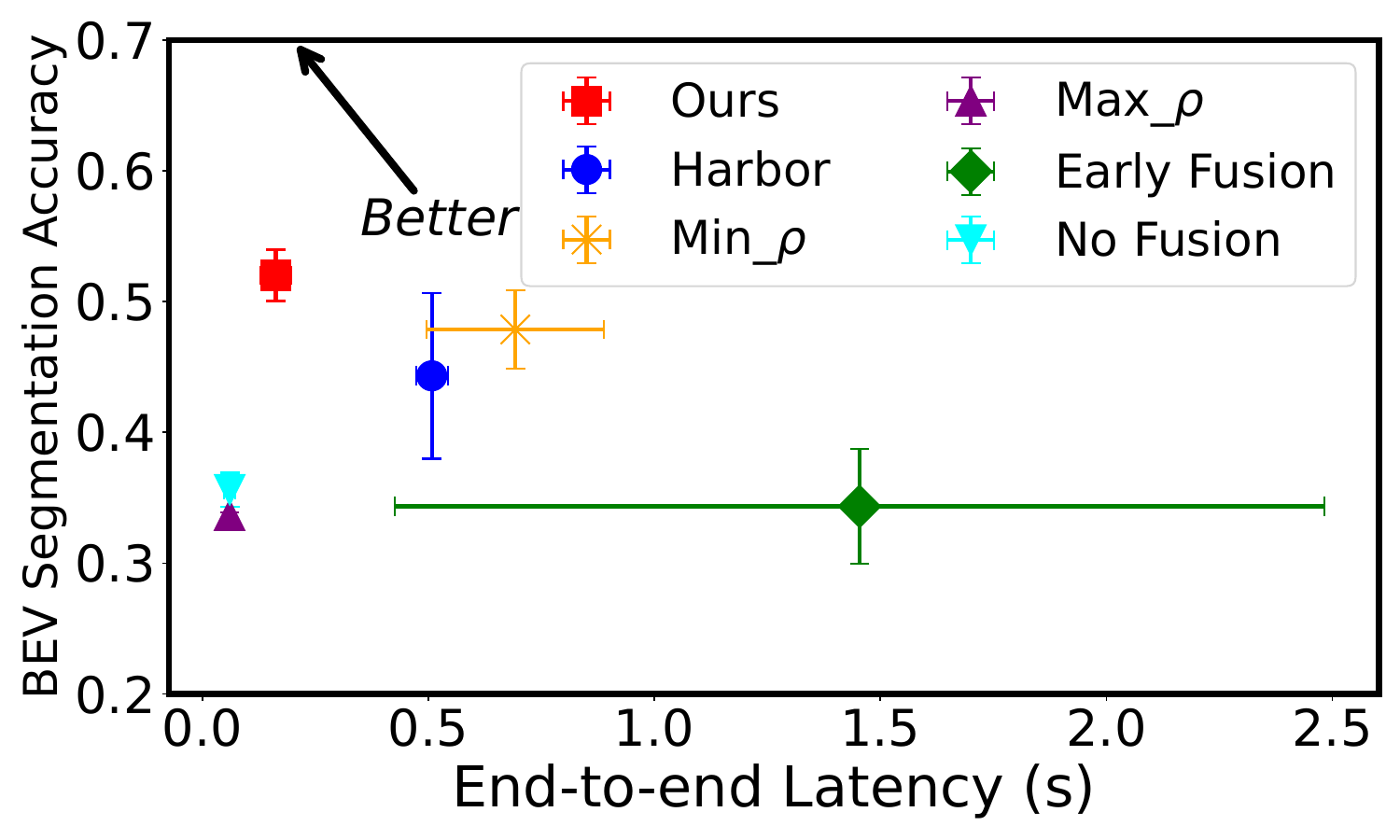}\label{fig:sim2_1}}
            \hspace{-0.05in}
        		\subfigure[Low V2V throughput]{
		\includegraphics[width=0.23\textwidth]{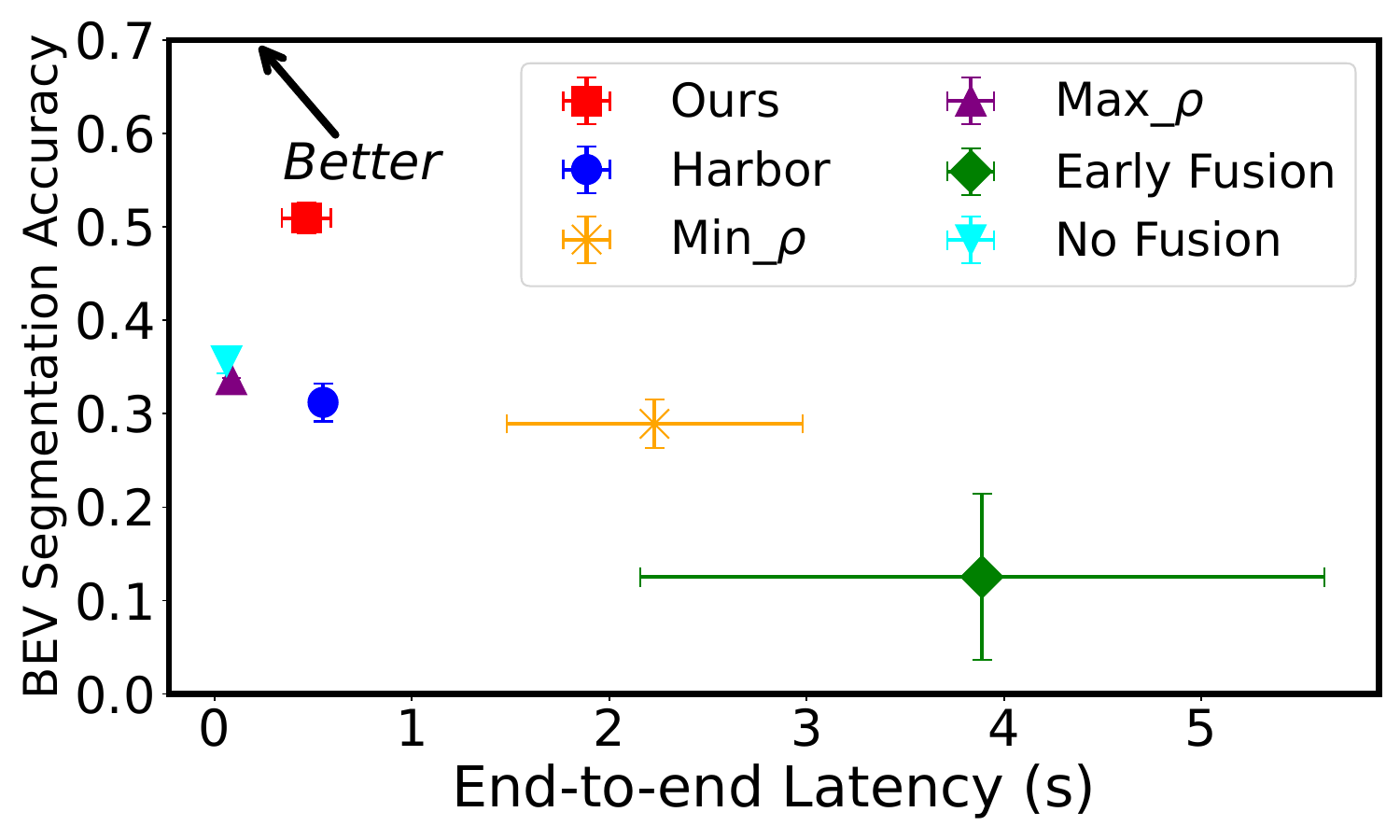}\label{fig:sim2_2}}
         \hspace{-0.05in}
        		\subfigure[The number of straggler CAVs]{
		\includegraphics[width=0.24\textwidth]{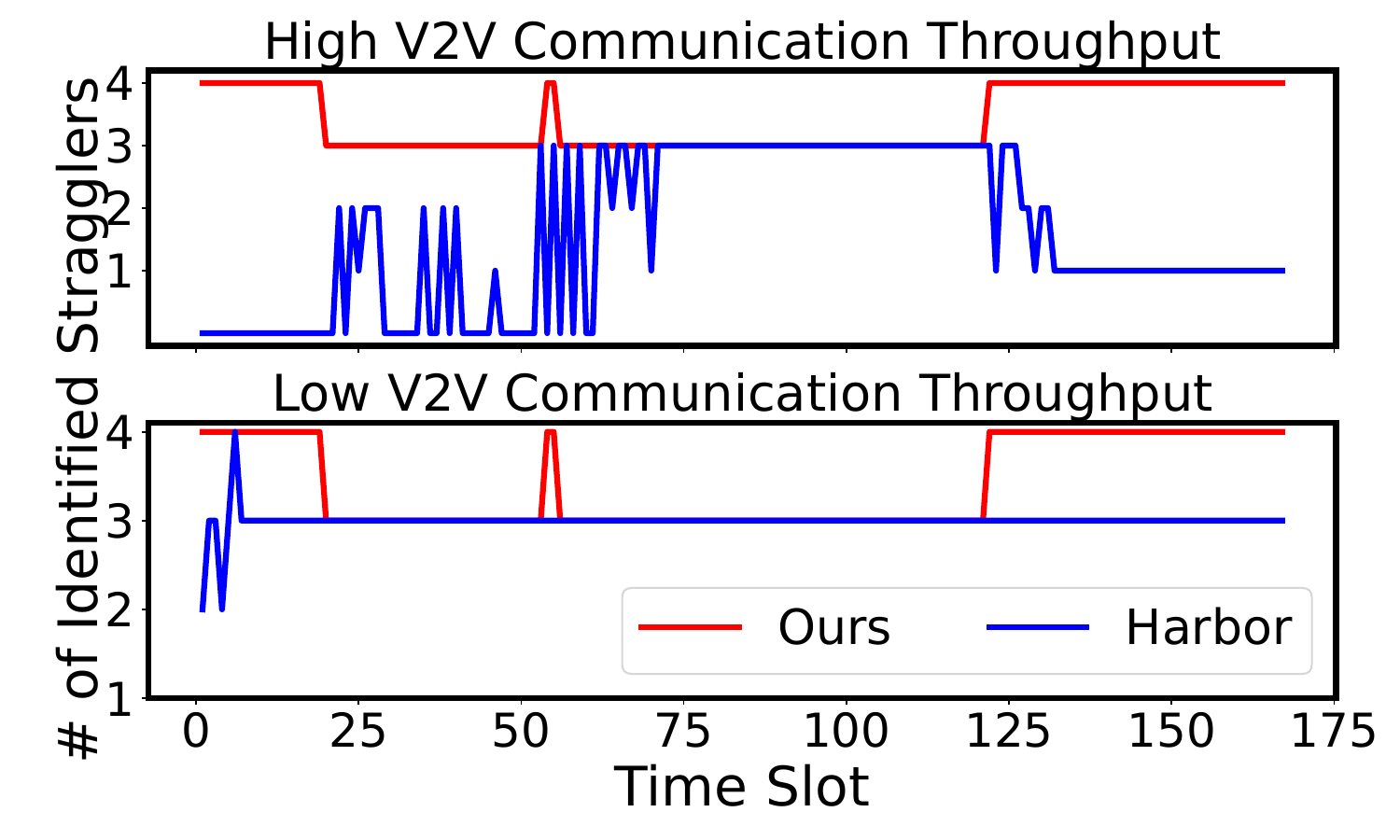}\label{fig:sim2_3}}
          \hspace{-0.05in}
        		\subfigure[Fusion deadlines]{
		\includegraphics[width=0.23\textwidth]{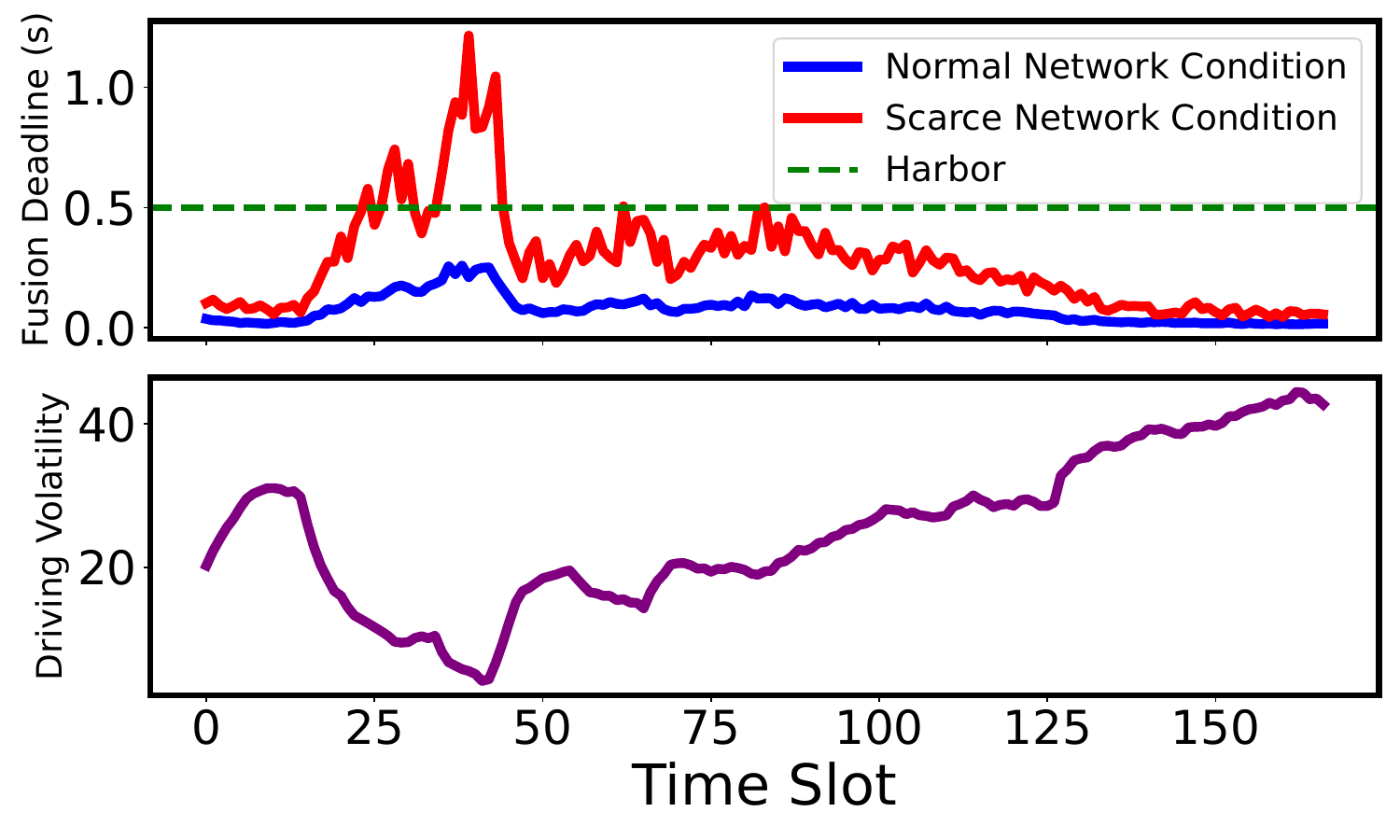}\label{fig:sim2_4}}
                \vspace{-0.1in}
         \caption{Comparative analysis of different baselines on end-to-end collaborative perception performance.}
		\label{fig:sim2}
                \vspace{-0.28in}
\end{figure*}

\section{Performance Evaluation}
\subsection{Experiment Setup}
\textbf{Dataset and backbones.} Our experiments are conducted on OPV2V \cite{opv2v}, which is a large-scale simulated dataset tailored for V2V collaborative perception. We select three representative urban intersection scenarios with $1-4$ collaborative CAVs and $20-40$ unconnected vehicles (objects). When computing the marginal BEV contribution, each CAV's FoV is defined as a $100 m\times100 m$ rectangular area. We consider two V2V sidelink throughput scenarios \cite{sidelink} in which the ego CAV is allocated bandwidth resources of 15–25 Mbps (low throughput) and 40–50 Mbps (high throughput), respectively. Beyond the backbones introduced in Section \ref{sec:motivation}, we use squeeze-and-excitation network \cite{compress} to performs channel-wise BEV feature compression. The compression ratio is selected from $\{1,2,4,8,16,32,64\}$ where a value of $1$ denotes no compression. The performance of BEV segmentation task is evaluated using the mIoU metric and average data size for camera image and BEV feature size is $2.46$ MB and $512.4$ KB per frame, respectively.

\textbf{Real-world testbeds.} To simulate different onboard computational capabilities, our method is deployed on both NVIDIA RTX 3080Ti and Jetson Orin hardware platform. The default values of $\omega, \alpha$ and $D$ are experimentally set to $1, 0.1$ and $0.5$, respectively. Performance analysis under varying parameter settings is also presented below. Through offline fitting, the value of $\beta$ and $\gamma$ in (\ref{eq:compensation}) are set to $0.34$ and $0.15$. In each time slot, the parameters $l_f^{min}$ and $l_f^{max}$ in (\ref{eq:fd}) are set to the time required by the CAV with the poorest bandwidth to transmit BEV features under the highest and lowest compression ratios.

\textbf{Benchmarks.} Our method is compared with the following benchmarks. (1) \textbf{ECOP} \cite{jiawei}: An upper-confidence-bound-based method that selects CAVs according to $\bar{g_i} + \sqrt{2\text{In}\,t/3\theta_{i,t}}$, where $\theta_{i,t}$ is the selection counter of CAV $i$ up to time $t$. (2) \textbf{MASS} \cite{mass}: It selects CAVs sequentially according to the sum $\bar{g}_i+0.6\sqrt{t-\tau_i}$, where $\tau_i$ is the last time ego CAV selects CAV $i$. (3) \textbf{Random}: A baseline that randomly selects collaborative CAVs without considering utility. (4) \textbf{Harbor} \cite{harbor}: Any CAV whose BEV data transmission latency exceeds $500$ ms is classified as a straggler by ego CAV and their perception data is discarded. (5) \textbf{Max$\boldsymbol{\_\rho}$} and \textbf{    Min$\boldsymbol{\_\rho}$}: These two methods constitute baseline approaches for BEV feature transmission by consistently applying the maximum and minimum compression ratios, respectively. (6) \textbf{Early Fusion} and \textbf{No fusion}: The ego CAV requests raw images from surrounding CAVs or constructs the BEV map in stand-alone manner.

\subsection{Performance Analysis}
\textbf{Superiority of CAV selection.} We begin by evaluating the performance gap between various collaborative CAV selection methods and the \textbf{Optimal} approach, which has hindsight information and always selects the top-$K$ CAVs. As shown in Figure \ref{fig:sim1_1} and \ref{fig:sim1_3}, our method achieves the smallest average gap across three driving scenarios. This improvement stems from the alternating exploration–exploitation mechanism in Algorithm \ref{alg:aecs}, which enables the ego CAV to effectively balance learning and exploitation. Figure \ref{fig:sim1_2} further shows that when the selection budget is limited to $K=1$, our method yields the most significant gains, outperforming ECOP and MASS by $46.6\%$ and $62.1\%$, respectively. This is because, with only a single selection opportunity, the ego CAV must quickly identify the most valuable collaborator. Our method addresses this by efficiently collecting marginal BEV contributions from all candidates. Furthermore, as shown in Figure \ref{fig:sim1_4}, increasing the value of $D$ extends the exploration phase, highlighting the algorithm's adaptability to varying driving environments through appropriate tuning of $D$.

\textbf{End-to-end performance comparison.} Next, we evaluate the BEV segmentation mIoU and average BEV feature transmission latency across different methods. As shown in Figure \ref{fig:sim2_1} and \ref{fig:sim2_2}, our method outperforms the benchmarks in both V2V networking conditions. In specific, compared with Harbor, our method achieves up to $67.9\%$ and $63.18\%$ end-to-end latency reduction and BEV segmentation accuracy improvement. This gain stems from dynamically evaluating driving volatility to set fusion deadlines as shown in Figure \ref{fig:sim2_4}, which in turn enables identification of a greater number of stragglers as shown in Figure \ref{fig:sim2_3}. Unlike Harbor, which uses a fixed deadline and discards delayed data, our approach employs adaptive feature compression, allowing the ego CAV to incorporate features from more collaborators. Additionally, compared to the Min$\_\rho$ baseline, our method constructs more accurate BEV maps, demonstrating the effectiveness of feature compression in mitigating the straggler effect.

\textbf{Decomposed time overhead on devices.} To further demonstrate the low computational overhead of our method, Figure \ref{fig:breakdown} presents the time overhead of the BEV‐feature‐based collaborative perception system across two distinct testbeds. As shown in Figure \ref{fig:breakdown1}, the average system inference latency on an RTX 3080Ti is $186.7$ ms. Communication latency dominates the total overhead, accounting for up to $71\%$, while Algorithm \ref{alg:aecs} contributes only $1.8\%$. According to Figure \ref{fig:breakdown2}, when the CAV’s onboard computing capability is limited, the time overhead associated with BEV feature extraction, compression, and fusion on Jetson Orin constitutes the primary bottleneck to the system’s real-time performance, accounting for $83.9\%$, while our method adds only $3.36\%$ overhead.

\textbf{Advantage of the volatility-aware feature fusion.} Finally, we visualize BEV maps at an intersection constructed by the ego CAV using our method and Harbor. As shown in Figure \ref{fig:illustration}, when V2V throughput to a collaborative CAV is low, both methods identify it as a straggler. Nonetheless, our method ensures timely BEV data delivery, yielding a more accurate and comprehensive map. This underscores its vital role in enhancing the accuracy of autonomous driving systems.

 \begin{figure}[t]
            \centering
            	\subfigure[On RTX 3080Ti]{
		 \includegraphics[width=0.45\columnwidth]{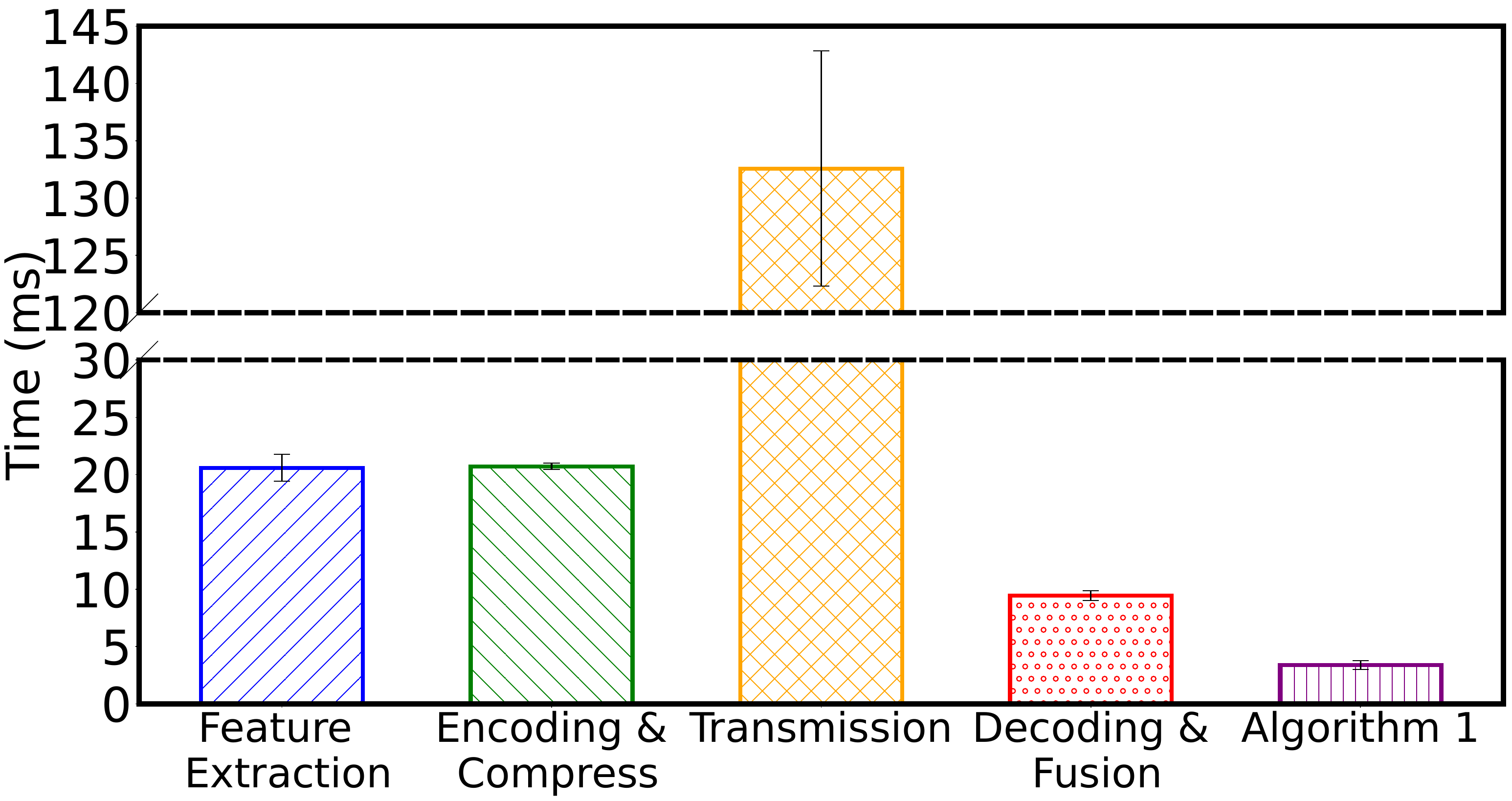}\label{fig:breakdown1}}
            \hspace{-0.05in}
            \subfigure[On Jetson Orin]{
		 \includegraphics[width=0.45\columnwidth]{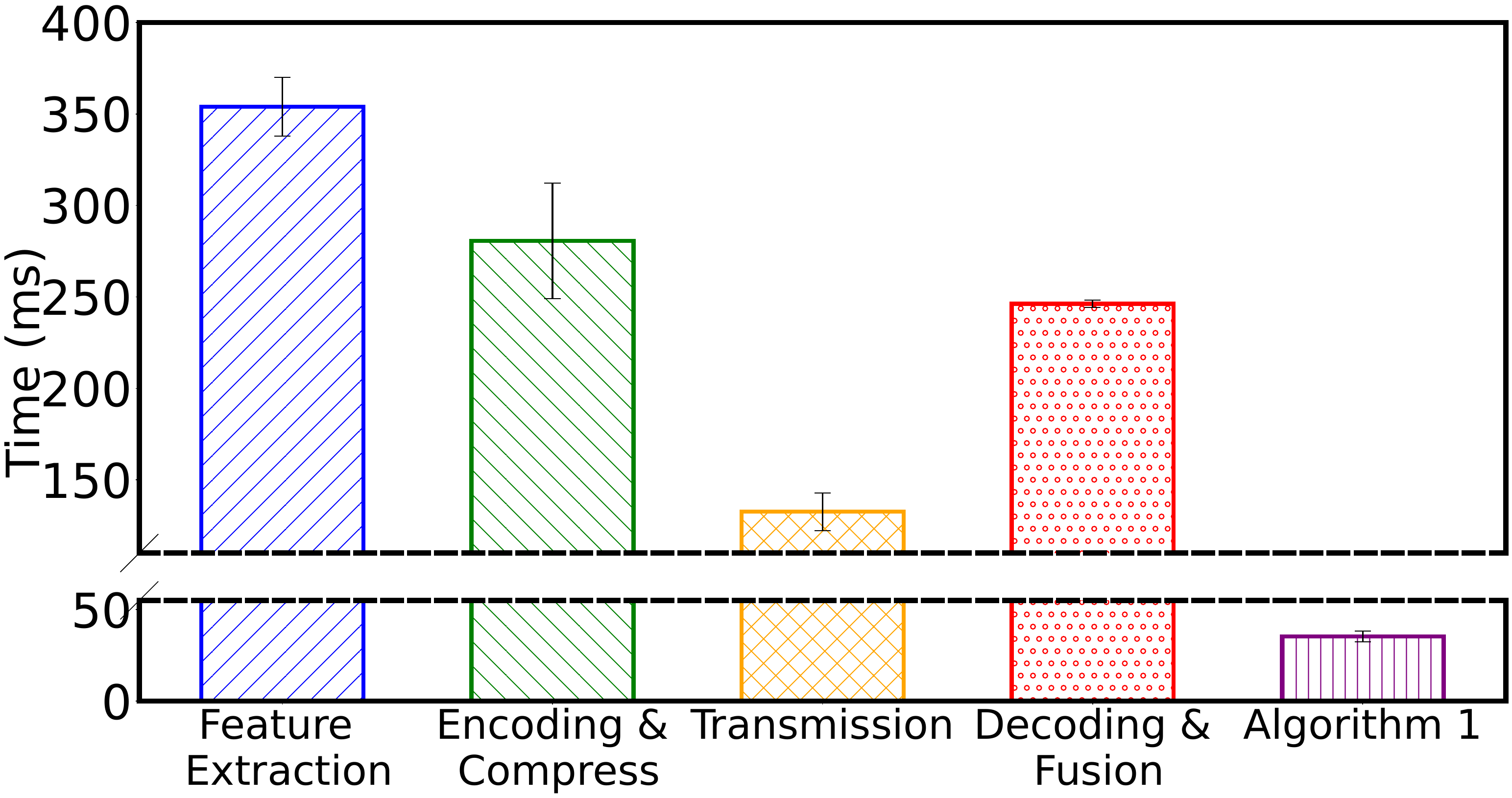}\label{fig:breakdown2}}
                         \vspace{-0.1in}
            \caption{Illustration of system latency breakdown on different devices. }\label{fig:breakdown}
                    \vspace{-0.28in}
        \end{figure}

	\section{Related Works}
	\label{sec:related}
    \textbf{CAV Selection in Collaborative Perception.} Prior studies have investigated strategies for selecting collaboration partners in V2V-based perception. Who2com \cite{liu2020who2com} introduced a three-stage handshaking mechanism, which enables targeted collaboration but incurs substantial round-trip latency. To reduce communication overhead, methods such as V2VNet \cite{v2vnet}, CPIM \cite{zhang2025}, and Where2com \cite{hu2022where2comm} allow CAVs to broadcast perception data without explicit coordination. However, this broadcast paradigm becomes inefficient in dense traffic scenarios due to significant communication load. Recent works employ online learning-based CAV selection, enabling the ego CAV to gradually infer the utility of neighbors. For example, ECOP \cite{jiawei} and MASS \cite{mass} leverage multi-armed bandit (MAB) frameworks to select LiDAR-equipped CAVs based on historical confidence scores or detection accuracy.   \textit{However, these methods are not directly applicable to collaborative BEV perception, as they cannot effectively evaluate the utility of BEV features or balance the trade-off between exploration and exploitation of utilities in dynamic environments.}

    \textbf{Communication Optimizations in Collaborative Perception.} The communication latency significantly impacts collaborative perception gains, consequently, several studies \cite{luo2023edgecooper, pacp, cp1, cp2} consider the total V2V transmission latency as one of the constraints in their optimization problem. In addition, Harbor \cite{harbor} mitigates the straggler effect by discarding the perception data of CAVs whose transmission latency exceeds a predefined deadline. \textit{While effective, these methods account only for a fixed latency constraint or fusion deadline, overlooking the varying urgency of the ego CAV’s data requirements under different driving conditions.}

	\section{Conclusion}
	\label{sec:conclusion}
    In this work, we have proposed BEVCooper. Through preliminary studies, we have developed a novel BEV feature evaluation metric and identified key challenges and potential directions for improving perception accuracy in dynamic driving environments. Building upon these insights, BEVCooper incorporates an online learning based CAV selection strategy that balances exploration and exploitation to maximize efficiency and accommodate promising but underutilized vehicles.  To mitigate straggler effect, BEVCooper designs a volatility-aware fusion mechanism that adapts to environmental dynamics and V2V link quality. The effectiveness of BEVCooper has been validated through implementation on two real-world testbeds and experiments across diverse scenarios. We believe that the design of BEVCooper plays a critical role in enhancing the safety and robustness of autonomous driving systems.

\begin{comment}
    \section{Acknowledgment}
The work of Jiawei Hou and Peng Yang is supported by the Natural Science Foundation of China
under Grant 62001180 and the Young Elite Scientists Sponsorship
Program by CAST under Grant 2022QNRC001. The work of Xiangxiang Dai is supported by the National Natural Science Foundation of China under Grant 625B2163.
\end{comment}

 \begin{figure}[t]
            \centering
            \includegraphics[width=0.8\columnwidth]{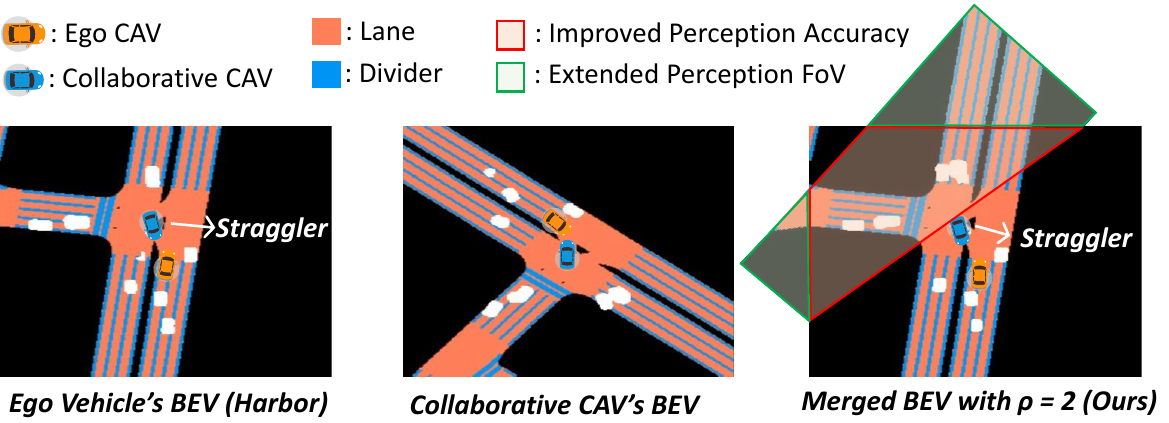}                         \vspace{-0.1in}
            \caption{Performance advantage of our method compared to Harbor.}
            \label{fig:illustration}
                    \vspace{-0.28in}
        \end{figure}

    \begin{comment}
        \begin{align}
         Reg(T) &= T\sum\limits_{i=1}^K\mu_{\delta_i}-\mathbb{E}\left[\sum\limits_{t=1}^T\sum\limits_{i=1}^Na_i(t)g_i(t)\right]\\
         & = \underbrace{T\sum\limits_{i=1}^K\mu_{\delta_i}-\sum\limits_{i=1}^N\mu_{i}Q_i(t)}_{\text{(a)}}+\label{eq:decompose_reg}\\&\notag\underbrace{\sum\limits_{i=1}^N\mu_{i}Q_i(t)-\mathbb{E}\left[\sum\limits_{t=1}^T\sum\limits_{i=1}^Na_i(t)g_i(t)\right]}_{\text{(b)}}
    \end{align}
    \end{comment}
    
\begin{appendices}
\section{Proof of Theorem 1}
\label{sec:appendix1}
The cumulative regret in (\ref{eq:reg}) is firstly decomposed into two distinct terms, each corresponding to a different source of suboptimal collaborative CAV selection:

\resizebox{0.95\linewidth}{!}{
\begin{minipage}{\linewidth}
\begin{align}
\label{eq:decompose_reg}
\underbrace{T\sum\limits_{i=1}^K\mu_{\delta_i}-\sum\limits_{i=1}^N\mu_{i}Q_i(t)}_{\text{(a)}} 
          +\underbrace{\sum\limits_{i=1}^N\mu_{i}Q_i(t)-\mathbb{E}\left[\sum\limits_{t=1}^T\sum\limits_{i=1}^Na_i(t)g_i(t)\right]}_{\text{(b)}}
\end{align}
\end{minipage}}
    where $Q_i(t)=\mathbb{E}\left[\sum_{s=1}^ta_i(s)\right]$ is the expected number of times CAV $i$ is selected up to time $t$. The term $\sum_{i=1}^N\mu_{i}Q_i(t)$ calculates the stationary rewards according to Algorithm \ref{alg:aecs}.
    Equation (\ref{eq:decompose_reg}) exhibits that the cumulative regret of Algorithm \ref{alg:aecs} can be decomposes into two parts: (a) the regret incurred by selecting CAVs whose expected reward is suboptimal; (b) the mismatch between the expected stationary rewards and the actual rewards obtained. The discrepancy in (b) results from the intermittent selection of CAVs in the restless environment, which causes their state distributions at selection times to deviate from the stationary distribution. In what follows we bound the two parts separately.

\subsection{Bounding part (b) in (\ref{eq:decompose_reg})}
According to Lemma 1 from \cite{restless3} and Lemma 2.1 from \cite{restless2}, we use the following results on Markov chain.
     \begin{lemma}
    \label{lemma1}  
        Assume the hidden Markov transition model of each CAV's reward (marginal BEV contribution) irreducible and aperiodic. If CAV $i$ is selected for $Q_i$ consecutive time steps. then its cumulative expected contribution satisfies: $\mathbb{E}\left[\sum_{t=1}^{T}g_i(t)-\mu_iQ_i\right]\le C_P$, where $C_P$ is a constant only depends on the transition model and is independent of $Q_i$.
    \end{lemma}

 Lemma \ref{lemma1} establishes that in a Markovian reward setting, frequent switching between CAVs interrupts the convergence of their internal state processes to the stationary distribution. As a result, when a CAV is selected after a period of inactivity, its reward is drawn from a transient distribution, incurring a regret bounded by a constant $C_P$. Therefore, bounding part (b) in (\ref{eq:decompose_reg}) reduces to analyzing the number of CAV switches under the alternating exploration–exploitation structure.

We begin by analyzing the exploration phase. Since the duration of each exploration epoch grows geometrically with base 2, starting from $\lceil N/K \rceil$, the total time spent on CAV $i$ during exploration up to time $t$ is upper bounded by $\frac{2^{O_i+1} - 2}{K} \le D \log_2 t$. Consequently, the number of exploration phases can be bounded by:
    \begin{equation}
    \label{eq:explor_num}
        O_i\le \log_2\left(\frac{KD\text{log}_2\,t + 2}{2}\right).
    \end{equation}
Consequently, the total number of switches involving the selected CAV set during exploration is at most $N\log_2\left(\frac{KD\text{log}_2\,t + 2}{2}\right)$. 

By time $t$, the total duration allocated to the exploitation phase is upper bounded by $t-N/K$, due to the initial exploration phase consuming at least $N/K$ time slots. Analogously, we have $2^{I_i+1}-1\le t-N/K$, which implies:
    \begin{equation}
    \label{eq:exploi_num}
        I_i \le \log_2\, \left( t-\frac{N}{K}+1\right).
    \end{equation}
Since each phase may incur up to $K$ CAV switches, the total number of switches involving the selected CAV set during exploitation is upper bounded by $K \log_2\left( t - \frac{N}{K} + 1 \right)$. Letting $\bar{C} = \max_{i=1}^{N} C_P$, the regret corresponding to part (b) in (\ref{eq:decompose_reg}) is bounded by:
\begin{comment}
    \begin{align}
    \label{eq:rega}
        &\sum\limits_{i=1}^N\mu_{i}Q_i(t)-\mathbb{E}\left[\sum\limits_{t=1}^T\sum\limits_{i=1}^Na_i(t)g_i(t)\right] \le \\&\notag \bar{C}\left[N\log_2\left(\frac{KD\text{log}_2\,t + 2}{2}\right)+K\log_2\, \left( t-\frac{N}{K}+1\right)\right],
    \end{align}
    \end{comment}
    \begin{align}
            \label{eq:rega}
\bar{C}\left[N\log_2\left(\frac{KD\text{log}_2\,t + 2}{2}\right)+K\log_2\, \left( t-\frac{N}{K}+1\right)\right],
    \end{align}
    where the right side of the inequality is a logarithmic function of time $t$.

    \subsection{Bounding part (a) in (\ref{eq:decompose_reg})}

    In this subsection, we further decompose the total expected number of suboptimal CAV selections into contributions from the exploration and exploitation phases. According to (\ref{eq:explor_num}), the accumulative number of time slots used to explore CAV $i$ by time $t$, denoted by $T_i$ is bounded by:
    \begin{equation}
        T_i \le  \frac{D}{2}\log_2t-\frac{1}{K}.
    \end{equation}
    
    In each exploration round, up to $K\lceil N/K \rceil$ CAVs, including both optimal and suboptimal ones, are selected. The regret incurred from selecting suboptimal CAVs in each round is:
    \begin{comment}
\begin{align}
       & \notag\left(\frac{D}{2}\log_2t-\frac{1}{K}\right) \left(\lceil \frac{N}{K} \rceil \sum\limits_{i=1}^{K}\mu_{\delta_i} - \sum\limits_{i=1}^{K\lceil \frac{N}{K} \rceil}\mu_{\delta_i}\right) \\
       & \le \left(\frac{D}{2}\log_2t-\frac{1}{K}\right)  \left(\lceil \frac{N}{K} \rceil \sum\limits_{i=1}^{K}\mu_{\delta_i} - \sum\limits_{i=1}^{N}\mu_{\delta_i}\right).
       \label{eq:regb1}
    \end{align}
    \end{comment}
    \begin{align}
        \left(\frac{D}{2}\log_2t-\frac{1}{K}\right)  \left(\lceil \frac{N}{K} \rceil \sum\limits_{i=1}^{K}\mu_{\delta_i} - \sum\limits_{i=1}^{N}\mu_{\delta_i}\right).
       \label{eq:regb1}
    \end{align}
     As stated in Theorem \ref{math:theorem1}, $D$ is a non-negative constant, then (\ref{eq:regb1}) increases with the same order of $\log_2t$. We now analyze the selection of suboptimal CAVs during exploitation phases. Although the exploitation phase is intended to utilize the best-performing CAVs based on empirical rewards, estimation noise, particularly under Markovian reward processes, may occasionally result in the selection of suboptimal CAVs.  Specifically, such errors occur when the empirical mean of a suboptimal CAV, e.g., $\bar{g}_i$, exceeds that of an $K$-optimal one, e.g., $\bar{g}_j$, at the beginning of an exploitation epoch. To quantify the probability of overestimation of CAV $i$ or underestimation of CAV $j$, we invoke the following lemma \cite{restless2}.
     \begin{lemma}
     \label{lemma2}
        Let arms $i$ and $j$ be governed by irreducible, aperiodic Markov chains on finite state spaces $\mathcal{S}_i$ and $\mathcal{S}_j$, with stationary distributions $\pi_i, \pi_j$, and stationary means $\mu_i < \mu_j$. Denote by $t_{I_i}$ the start time of the $I_i$-th exploitation phase, and let $\Omega_{[i,j, I_i]}$ be the event that arm $i$'s sample mean exceeds that of $j$ at exploitation phase $I_i$, then the probability of event $\Omega_{[i,j, I_i]}$ satisfies:
         %Let arm $i$ and arm $j$ be two Markovian arms with finite state spaces $\mathcal{S}_i$ and $\mathcal{S}_j,$ and stationary distributions $\pi_i, \pi_j$, respectively. Suppose that $\mu_j > \mu_i$, and let $t_{I_i}$ denote the starting time of $I_i$-th exploitation phase. Let $\Omega_{[i,j, I_i]}$ be the event that arm $i$'s sample mean exceeds that of $j$ at exploitation phase $I_i$, 
        \begin{align}
            \notag Pr[\Omega_{[i,j, I_i]}] \le 
             \sum_{n=i,j}\left(\frac{1}{\log2}+\frac{\sqrt{2}\xi_n}{\sum_{s\in\mathcal{S}_n}s}|\mathcal{S}_n|\right)\frac{1}{t_{I_i}\pi_{min}},
         \end{align}
        where $\xi_n$ depends on the initial reward distribution of each CAV's marginal BEV contribution and $\pi_{min}$ denotes the minimum value of stationary distribution probabilities across all collaborative CAVs and states.
     \end{lemma}
The cumulative regret resulting from the the selection of suboptimal CAVs during the exploitation phases is:
\begin{align}
    \quad I_i2^{I_i-1} \underbrace{\sum\limits_{j=1}^K\sum\limits_{i=K+1}^N(\mu_{\delta_{j}}-\mu_{\delta_{i}})}_{\Delta}Pr[\Omega_{[i,j, I_i]}],
    \label{eq:reg_exploitation}
\end{align}
where the term $\Delta$ bounds the regret from selecting one group of collaborative CAVs in exploitation phase. Then, according to Lemma \ref{lemma2}, (\ref{eq:reg_exploitation}) can be bounded by:
\resizebox{0.95\linewidth}{!}{
\begin{minipage}{\linewidth}\begin{align}
         &\notag\quad I_i2^{I_i-1} \Delta\sum_{n=i,j}\left(\frac{1}{\log2}+\frac{\sqrt{2}\xi_n}{\sum_{s\in\mathcal{S}_n}s}|\mathcal{S}_n|\right)\frac{1}{t_{I_i}\pi_{min}}\\&\le\frac{\Delta2^{I_i-1}}{t_{I_i}\pi_{min}} \log2\, \left( t-\frac{N}{K}+1\right)\sum_{n=i,j}\left(\frac{1}{\log2}+\frac{\sqrt{2}\xi_n}{\sum_{s\in\mathcal{S}_n}s}|\mathcal{S}_n|\right)\\
         &\label{eq:regb2} \le\frac{\Delta}{\pi_{min}} \log2\, \left( t-\frac{N}{K}+1\right)\sum_{n=i,j}\left(\frac{1}{\log2}+\frac{\sqrt{2}\xi_n}{\sum_{s\in\mathcal{S}_n}s}|\mathcal{S}_n|\right).
    \end{align}
\end{minipage}}
Inequality (\ref{eq:regb2}) holds since  $t_{I_i}\ge \lceil \frac{N}{K}\rceil + 2^{I_i}-2\ge2^{I_i-1}$ when $I_i\ge1$. Combining (\ref{eq:rega}), (\ref{eq:regb1}) and (\ref{eq:regb2}), it can be observed that both regret parts in (\ref{eq:decompose_reg}) grows logarithmically with time. This concludes the proof of Theorem \ref{math:theorem1}.

\end{appendices}
\clearpage
	\bibliographystyle{ieeetr}
	\bibliography{reference}{}

\end{document}